%% file: icml_new.tex
\documentclass{article}
\usepackage[latin9]{inputenc}
\usepackage{booktabs}
\usepackage{multirow}
\usepackage{amsmath}
\usepackage{amsthm}
\usepackage{amssymb}
\usepackage{graphicx}
\usepackage{wasysym}
\usepackage{microtype}
\usepackage[unicode=true,
 bookmarks=false,
 breaklinks=false,pdfborder={0 0 1},backref=section,colorlinks=false]
 {hyperref}
 \usepackage{todonotes}

\makeatletter

\providecommand{\tabularnewline}{\\}

\theoremstyle{definition}
 \newtheorem{example}{\protect\examplename}
\theoremstyle{plain}
\newtheorem{thm}{Theorem}[section]
\newtheorem{prop}[thm]{Proposition}

\theoremstyle{definition}
\newtheorem{defn}[thm]{Definition}

\theoremstyle{remark}
\newtheorem{rem}[thm]{Remark}

\usepackage{subfigure}
\usepackage[accepted]{icml2023}

\icmltitlerunning{Intrinsic Sliced Wasserstein Distances on Manifolds and Graphs}

\def\E{\mathbb{E}}
\def\P{\mathcal{P}}
\def\B{\mathcal{B}}
\def\R{\mathbb{R}}
\def\X{\mathcal{X}}
\def\D{\mathcal{D}}
\def\W{\mathcal{W}_2}
\def\H{\mathcal{H}}
\def\M{C}

\makeatother

\providecommand{\examplename}{Example}

\begin{document}
\twocolumn[
\icmltitle{Intrinsic Sliced Wasserstein Distances for Comparing Collections of\\ Probability Distributions on Manifolds and Graphs}

% It is OKAY to include author information, even for blind % submissions: the style file will automatically remove it for you % unless you've provided the [accepted] option to the icml2021 % package.
% List of affiliations: The first argument should be a (short) % identifier you will use later to specify author affiliations % Academic affiliations should list Department, University, City, Region, Country % Industry affiliations should list Company, City, Region, Country
% You can specify symbols, otherwise they are numbered in order. % Ideally, you should not use this facility. Affiliations will be numbered % in order of appearance and this is the preferred way. 
\icmlsetsymbol{equal}{*}
\begin{icmlauthorlist} 
% \icmlauthor{Aeiau Zzzz}{equal,to} 
% \icmlauthor{Bauiu C.~Yyyy}{equal,to,goo} 
\icmlauthor{Raif Rustamov}{amz,att} 
\icmlauthor{Subhabrata Majumdar}{sm,att} 
\end{icmlauthorlist}

\icmlaffiliation{amz}{Amazon, New York, NY, USA} 
\icmlaffiliation{sm}{AI Risk and Vulnerability Alliance, Seattle, WA, USA}
\icmlaffiliation{att}{work done while at AT\&T.}
\icmlcorrespondingauthor{Subhabrata Majumdar}{zoom.subha@gmail.com} 

% You may provide any keywords that you 
% find helpful for describing your paper; these are used to populate 
% the "keywords" metadata in the PDF but will not be shown in the document 
\icmlkeywords{Machine Learning, ICML}
\vskip 0.3in ]

% this must go after the closing bracket ] following \twocolumn[ ...
% This command actually creates the footnote in the first column 
% listing the affiliations and the copyright notice. 
% The command takes one argument, which is text to display at the start of the footnote. 
% The \icmlEqualContribution command is standard text for equal contribution. 
% Remove it (just {}) if you do not need this facility.
\printAffiliationsAndNotice{} 
% leave blank if no need to mention equal contribution \printAffiliationsAndNotice{\icmlEqualContribution} 
% otherwise use the standard text. 
\begin{abstract}
Collections of probability distributions arise in a variety of applications ranging from user activity pattern analysis to brain connectomics. In practice these distributions can be defined over diverse domain types including finite intervals, circles, cylinders, spheres, other manifolds, and graphs. This paper introduces an approach for detecting differences between two collections of distributions over such general domains. To this end, we propose the intrinsic slicing construction that yields a novel class of Wasserstein distances on manifolds and graphs. These distances are Hilbert embeddable, allowing us to reduce the distribution collection comparison problem to a more familiar mean testing problem in a Hilbert space. We provide two testing procedures one based on resampling and another on combining p-values from coordinate-wise tests. Our experiments in various synthetic and real data settings show that the resulting tests are powerful and the p-values are well-calibrated.
\end{abstract}

\section{Introduction}
Distributional data defined over general domains such as manifolds and graphs arise in a variety
of statistical applications. 
% For instance, when analyzing 24-hour
% activity patterns by constructing histograms of activity counts by
% time, the resulting histograms are really supported on a circle rather
% than an interval. If in addition to the time of activity, the observations
% come with a real number such as the intensity of the activity, then
% we end up with a histogram over a cylindrical domain. Spatial datasets
% recorded at some geographic region level are another example: one
% can build a distribution over the region adjacency graph by capturing
% the normalized counts of events %(e.g. certain type of crime) 
% in each region. 
% When analyzing distributions
% over such general domains it is desirable to rely on methods that
% take into account the connectivity and geometry of the underlying
% domain, respect the distributional nature of the data, and lead to
% efficient practical algorithms.
In this paper we consider the problem of comparing two collections
of distributions over such a general domain. Our goal is to test for homogeneity---whether all of
the distributions come from the same \emph{meta-distribution}---in an interpretable manner. While
conceptually similar to two-sample testing, this is a higher order
notion in the sense that our units of analysis are distributions/histograms.

For instance, given collections of personal activity histograms (over cylinder: time of day $\times$ intensity) for
two sub-populations, one may be interested in determinining whether
there are statistically significant differences between activity patterns
of these sub-populations. As another example, consider normalized
counts of events  per geographic region on a daily
basis. Collected over a year, this gives a set of 365 daily probability
distributions over the region adjacency graph, and one may wish to compare the
collection of distributions from weekdays to those from weekends.
%can be tested or one rely on unspecific tests. In practice specific
%tests 
%As typical with two-sample tests, 
Testing specific aspects of homogeneity is preferable: for example, in regular two-sample testing, detecting that the means are unequal provides interpretable insights,
whereas a general test that only says there are unspecified differences
between the distributions is less useful for interpretation.

Limited settings of this problem tackling distributions over the interval/circle
have been considered in the literature \cite{10.1093/biomet/asz052}, yet the general case of distributions over
graphs and manifolds is open. The requirement to test for specific
differences is non-trivial on general domains: what is the equivalent
of mean for a collection of distributions? While Fr\'echet mean \cite{OTBook} may
seem like the natural choice, there are a number of problems with testing
Fr\'echet mean equality. First, the existence and uniqueness of the
Fr\'echet mean is not guaranteed, and it can be sensitive to small changes.
Second, computing the Fr\'echet mean is expensive and
can become prohibitive when resampling is used to compute the null
distribution. Finally, resampling poses conceptual problems: using
permutation null will detect differences beyond the equality of Fr\'echet
means (same problem exists in regular two-sample testing, see e.g.
\citet{permuteornot}), and using bootstrap requires designing
the null case, which is highly non-trivial.

We attack this problem using insights from recent developments that
utilize \emph{Hilbert embeddings} for simplifying distributional data
problems \cite{DBLP:conf/icml/SolomonRGB14,densityHilbert}.
The simplification comes as a result of linearity of Hilbert spaces,
which allows adapting existing statistical approaches 
%such as functional data methodology 
to distributional data. A crucial requirement on
the embedding is that the distance in the embedding space should give
a meaningful distance between measures; it is this property that renders
quantities computed in the embedding space such as means and variances
meaningful. While transportation based distances are efficient at
capturing many aspects of distributional data such as horizontal variation
\cite{wass_review1,OTBook,wass_review2}, yet the transportation theoretic
approaches hit a roadblock beyond the real line case due to their
Hilbert non-embeddability \cite{OTBook}. 

\begin{figure}
\centering{}\includegraphics[width=\columnwidth]{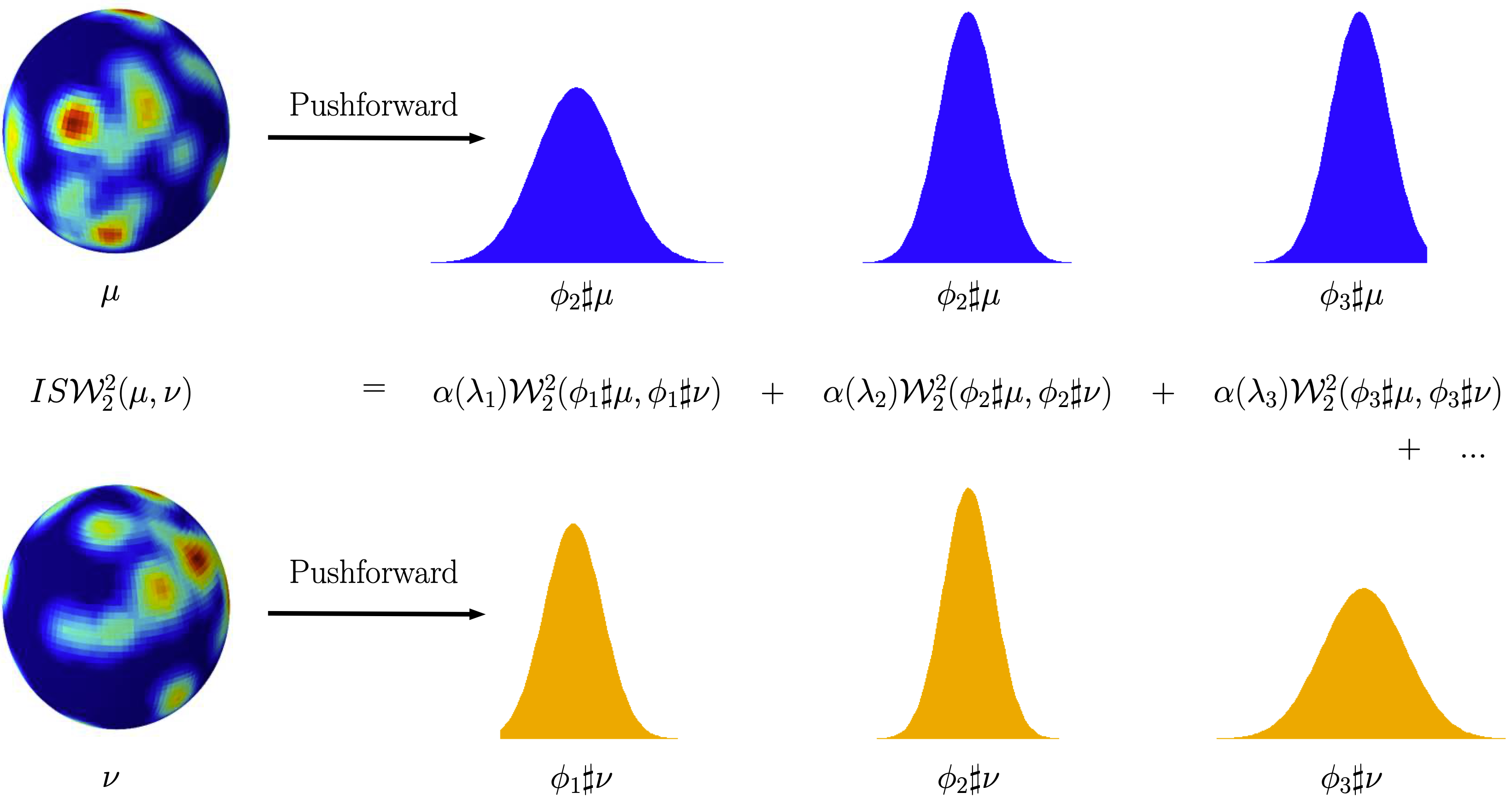}\caption{\label{fig:pushforward}Schematic of the proposed intrinsic slicing
construction. Given two probability measures on the sphere (here the
darkest blue corresponds to zero mass), different aspects of their
dissimilarities become apparent after pushforward to the real line
using the eigenfunctions of the Laplace-Beltrami operator, $\{\phi_{i}\}$,
in this case spherical harmonics. As a particular example of our general
construction, the (squared) intrinsic sliced 2-Wasserstein distance
$IS\W^{2}(\cdot,\cdot)$ is the weighted sum of the dissimilarities
of the corresponding pushforwards of $\mu$ and $\nu$ as measured
by squared 2-Wasserstein distance $\W^{2}(\cdot,\cdot)$ \emph{on
the real line}.}
\end{figure} 

We overcome these difficulties by introducing a new slicing construction
on manifolds and graphs (Figure \ref{fig:pushforward}) inspired by
the sliced 2-Wasserstein ($\W$) distances in high dimensional spaces \cite{Sliced1,sliced2}.
Our construction leverages eigenvalues and eigenfunctions/eigenvectors of the Laplace-Beltrami operator on manifolds and Laplacian matrix on graphs to capture the intrinsic
geometry and connectivity of the data domain. We apply this slicing
construction to obtain a novel class of intrinsic sliced 2-Wasserstein
distances on manifolds and graphs. The resulting distances are Hilbert
embeddable, have a number of desirable properties, and can be truncated
to obtain finite-dimensional embeddings.

Using the corresponding embedding allows us to reduce the distribution collection comparison problem to
the comparison of means in a high-dimensional Euclidean space. At the theoretical level our test checks equality of Fr\'echet means along slicings (see discussion after Example \ref{ex:1dim} in Section \ref{subsec:hilb_emb}). These means are transparently tied to the input data, whereby rejections lead to interpretable insights. We provide two approaches for hypothesis testing and verify via extensive
experiments %on synthetic and real data examples in a variety of data settings 
that these tests are powerful, and the $p$-values are well-calibrated.

\section{Related Work}
%\label{sec:related}
% \paragraph{Related work}
Our framework is not simply a higher order version of a two-sample kernel test \cite{mmd} since we test for equality of a specific aspect of meta-distributions. This renders our null hypothesis different from \citet{mmd}, and we need a different set of techniques both for proofs and computations. For example, testing in \citet{mmd} can use the permutation null which is valid due to the stronger null hypothesis of equal distributions. In contrast, with our null hypothesis we cannot use the permutation null and have to resort to a bootstrap procedure. 
%For more details see FAQ in Appendix \ref{sec:appC}.
Other approaches such as the general Hilbert embedding framework of \citet{densityHilbert}
is not tied to a distance between probability distributions and so
can be problematic for capturing the location and variability aspects
of distribution collections. In addition, \citet{densityHilbert} has
difficulties in higher dimensions and does not provide constructions
suitable for manifolds or graphs.

\paragraph{Sliced Distances}
The Sliced Wasserstein (SW) distance and its generalized variant \citep[GSW]{sliced2} sets up the idea of approximating Wasserstein distances using multiple nonlinear projections, it is presented in extrinsic terms (i.e. Euclidean space) and can suffer
from the curse of dimensionality when a low dimensional data manifold
lives in a high-dimensional space. Our choice of eigenfunctions for
projection is very different from the one-parameter function families
in GSW and allows us to rigorously prove a number of general and testing-specific properties. 

Moreover, the GSW construction does not directly apply to
graphs. While the tree-sliced variant of GSW \cite{treesliced}
can be applied in an intrinsic manner (the clustering variant), it
relies on a different type of distance, in the limit related to the
euclidean/geodesic distance. This can be seen by comparing our lower
bound to theirs: our lower bound for ISW is in terms of the MMD using
the spectral distance (Proposition \ref{prop:MMD-equiv}). The recently-proposed Sobolev transport \cite{sobolev} does consider measures supported on graphs, but in absence of a slicing construction it is limited to testing for unspecific differences and requires extensive compute due to much slower permutation based calibration (see Section~\ref{subsec:Simulations}). \citet{maxsliced} proposed Max Sliced Wasserstein (MSW) distance---taking maximum over projected $\W$ distances vs the average $\W$ distances of SW---in the context of generative models.

Finally, the robust sliced Wasserstein distance of \citet{LaiZhao} does make
use of the geometric properties of the underlying manifold. However,
their goal is to compute a correspondence between two manifolds by mapping them into $\mathbb{R}^{d}$ using eigenmaps and treating the mapped manifolds as measures in $\mathbb{R}^{d}$ and minimizing some version of Euclidean slicing. None of the aforementioned works consider the problem of comparing collections of distributions, provide ways for obtaining calibrated $p$-values, or prove properties of the distances that make them desirable for hypothesis testing.

\section{\label{sec:prelims}Preliminaries}
\subsection{Problem Setup}
Given a compact metric space $\X$, let $\P(\X)$ denote the set of
Borel probability measures on $\X$. Our main interest is in the case
where $\X$ is a graph or a manifold with the shortest/geodesic distance
as the metric, and thus the compactness restriction.  The 2-Wasserstein
distance can be defined on $\P(\X)$ using the metric of $\X$ as
the ground distance \cite{OTBook,wass_review1}, giving $\W^{\X}:\P(\X)\times\P(\X)\rightarrow\R_{\geq0}$. Due to the repeated use of the real line case we use the shorthand $\W$ when $\X=\R$, i.e. $\W:=\W^{\R}$. Central to our study are distributions on the space
of probability measures $\P(\P(\X))=(\P(\X),\B(\P(\X)))$, where $\B(\P(X))$
is the Borel $\sigma$-algebra generated by the topology induced by
$\W^{\X}$ \cite{wass_review2}. To avoid confusion, we will refer
to the elements of $\P(\P(\X))$ as \emph{meta-distributions}. 

Let $P,Q\in\P(\P(\X))$, and assume that we are given two collections
of probability measures $\{\mu_{i}\}_{i=1}^{N_{1}}$ and $\{\nu_{i}\}_{i=1}^{N_{2}}$
that are drawn from ${P}$ and ${Q}$: $\mu_{i}\sim{P}$ and $\nu_{i}\sim{Q}$
in an independent-and-identically-distributed (hereafter i.i.d.) manner.
Our goal is to use this sample to test the null hypothesis of whether
$P=Q$. While this is conceptually a two-sample test, note that our
data points are distributions; in practice, the distributions $\mu_{i}$
or $\nu_{i}$ are given by histograms. 

\begin{rem}
Let us compare this with the usual two-sample testing. Consider $P\in\P(\P(\X))$
constructed as follows. Let $\mu^{*}\in\P(\X)$ be a fixed probability
measure. Let $x_{1},x_{2},...x_{A}\sim\mu^{*}$ and construct the
histogram summarizing this sample: $\frac{1}{A}\sum_{a=1}^{A}\delta_{x_{a}}$.
Now, $\frac{1}{A}\sum_{a=1}^{A}\delta_{x_{a}}\in\P(\X)$ is one sample
drawn from $P$. Suppose one gets the collection $\{\mu_{i}\}_{i=1}^{N_{1}}$,
where each histogram is obtained as above: $\mu_{i}\sim P$. Similarly,
consider $Q\in\P(\P(\X))$ of the same type based on some other fixed
$\nu^{*}\in\P(\X)$, and let $\{\nu_{i}\}_{i=1}^{N_{2}}$ the corresponding
collection of histograms. Testing whether $P=Q$ in the limit boils
down to $\mu^{*}=\nu^{*}$. When compared to the usual two-sample
testing this may seem rather inefficient, requiring $A$ times more
samples (resp. $N_{1}A$ and $N_{2}A$ samples from $\mu^{*}$ and
$\nu^{*}$). However, in our general setup it is\emph{ not assumed} that the
histograms in the collections come from meta-distributions of the
above simple type (i.e. all $\mu_{i}$ are generated by drawing from
the same underlying distribution $\mu^{*}$). In fact, the target
use-case for our approach is when these histograms are collected by
observing different individuals who have their \emph{person-specific}
behaviors/distributions.
\end{rem}

\begin{rem}
To gain further insight into our problem setup, consider it through the lens of statistical manifold theory \citep[SMT]{smt}. Assume that we observe two collections of parametric distributions: $\mu_i=F(x, \beta_i), \beta_i \sim G(\beta)$ for $i=1,...,N_1$, and $\nu_i=F(x, \gamma_i), \gamma_i \sim H(\gamma)$ for $i=1,...,N_2$. Thus, each collection is generated from a distribution of its parameters, resp. $G(\cdot)$ and $H(\cdot)$. The goal of our test is to find out whether $G$ and $H$ are the same solely by observing the collections $\{\mu_{i}\}_{i=1}^{N_{1}}$ and $\{\nu_{i}\}_{i=1}^{N_{2}}$. While SMT and the Fisher information metric provides a powerful framework for studying parametric families of probability distributions, we focus on more general spaces of probability measures and the Fisher metric may not be easily extended to such nonparametric spaces of measures.
\end{rem}

\subsection{\label{subsec:hilb_emb}Hilbert Embeddings}

Let $\mathcal{D}(\cdot,\cdot):\P(\X)\times\P(\X)\rightarrow\mathbb{R}_{\geq0}$
be a distance between probability distributions. $\mathcal{D}(\cdot,\cdot)$
is called \emph{Hilbertian} (this is just a naming convention; no implication that the map is a Hilbert map) if there exist a Hilbert space $\H$ and
a map $\eta:\P(\X)\rightarrow\mathcal{\H}$ such that $\D(\mu,\nu)=\Vert\eta(\mu)-\eta(\nu)\Vert_{\mathcal{\H}}$.
For example, it is well-known that 2-Wasserstein distance on $\X=\mathbb{R}$
is Hilbertian \cite{OTBook} (also see Section \ref{subsec:Approximate-Hilbert-Embedding})
and Maximum Mean Discrepancy (MMD) on any $\X$ is Hilbertian \cite{mmd};
however, the 2-Wasserstein distance $\W^{\X}$ on general $\X$ is
not Hilbertian \cite{OTBook}.

Since the map $\eta$ takes every measure on $\X$ to a point in $\H$,
we see that a meta-distribution $P\in\P(\P(\X))$ gives a rise to a measure
on $\H$ given by pushforward operation, $\eta\#P=P\circ\eta^{-1}\in\P(\H)$.
In addition, if a finite dimensional approximation $\eta_{D}:\P(\X)\rightarrow\R^{D}$
of $\eta$ is available, then $\eta_{D}\#P$ is a measure on $\R^{D}$.
This observation is enormously useful: problems about the elements
of the rather abstract space $\P(\P(\X))$ are reduced to problems
about familiar measures on $\H$ or even $\R^{D}$. For example, the
usual notions of mean and variance can be applied to the measure $\eta\#P$
to gain insights about the meta-distribution $P$. The validity of
these insights hinges on the $\eta$-map coming from a Hilbertian
distance, as distances are central to the statistical quantities of
interest.

Testing for $\eta\#P=\eta\#Q$ can serve as a proxy for our original
testing problem of $P=Q$. As typical with two-sample tests, various
aspects of the equality $\eta\#P=\eta\#Q$ can be tested, such as
the mean or variance equality; unspecific tests of equality can be
applied as well. We will concentrate on testing certain aspects of
the equality so that one can easily drill down on the results. This
is similar to the regular two-sample testing where checking for equality
of, say, means is often preferable as it gives immediately interpretable
insights, whereas a general test that only says there are unspecified
differences between the distributions is less useful for interpretation.
%Our focus is in line with the recent surge of interest in interpretable
%learning. 
% \begin{rem}
% If an unspecific test is desired, Energy Distance based two sample
% test \cite{SZEKELY200558} or the MMD-based test \cite{mmd} can be
% adapted by replacing the Euclidean distance between points by a Hilbertian
% distance $\D$ between the distributions \cite{OTBook}. These tests
% can be calibrated using the permutation null and will not be discussed
% further.\qed
% \end{rem}
To obtain succint and interpretable tests we concentrate on the mean
of the resulting pushforward measure in $\H$. 
\begin{defn}
For a meta-distribution $P\in\P(\P(\X))$, define its \emph{Hilbert
centroid} with respect to the Hilbertian distance $\D$ as $\M_{\eta\#P}:=\E_{\mu\sim P}[\eta(\mu)]\in\H,$
assuming it exists.
\end{defn}
Our testing procedure is based on checking the equality $\M_{\eta\#P}=\M_{\eta\#Q}$,
or more explicitly: $\E_{\mu\sim{P}}[\eta(\mu)]=\E_{\nu\sim{Q}}[\eta(\nu)]$.
Intuitively, each ``dimension'' of the map $\eta$ probes some aspect
of the two involved meta-distributions and makes sure that they are
in agreement in expectation. One of our testing approaches will use
the statistic 
\begin{equation}
\mathbb{T}({P},{Q}):=\Vert\M_{\eta\#P}-\M_{\eta\#Q}\Vert_{\H}^{2}.
\label{eq:defn}
\end{equation}
to capture the deviations from equality; this quantity can be written
directly in terms of pairwise distances. 
\begin{prop}
\label{prop-defn}For $P,Q\in\P(\P(\X))$, the following holds:
\begin{align*}
\mathbb{T}({P},{Q}) &=
\E_{\mu\sim{P},\nu\sim{Q}}[\D^{2}(\mu,\nu)] - \\
& \frac{1}{2}\E_{\mu,\mu'\sim{P}}[\D^{2}(\mu,\mu')]-\frac{1}{2}\E_{\nu,\nu'\sim{Q}}[\D^{2}(\nu,\nu')].
\end{align*}
\end{prop}
Next we give an example of what Hilbert centroid equality implies in
an important special case.
\begin{example}
\label{ex:1dim}
Let $\X=[0,T]\subset\R$ with $\D$ being the 2-Wasserstein distance
$\W$. Given a probability measure $\mu\in\mathrm{\P}([0,T])$, let
$F_{\mu}$ be its cumulative distribution function: $F_{\mu}(x)=\mu([0,x])=\int_{0}^{x}d\mu$.
The generalized inverse of cumulative distribution function (CDF)
is defined by $F_{\mu}^{-1}(s):=\inf\{x\in[0,T]:F_{\mu}(x)>s\}$.
The squared 2-Wasserstein distance has a rather simple expression
in terms of the inverse CDF \cite{OTBook}:
\begin{equation}
\W^{2}(\mu,\nu)=\int_{0}^{1}(F_{\mu}^{-1}(s)-F_{\nu}^{-1}(s))^{2}ds.\label{eq:w2-via-quantiles}
\end{equation}
This formula immediately establishes the Hilbertianity of $\W$ through
the map $\eta:\P([0,T])\rightarrow L_{2}([0,T])$ defined by $\eta(\mu)=F_{\mu}^{-1}$.
Note that $\eta$ is invertible for increasing normalized functions
in the embedding space. Using this insight, we see that the corresponding
``average measure'' of $P\in\P(\P(\X))$ can be introduced via ${P}_{\mathrm{av}}=\eta^{-1}(\E_{\mu\sim{P}}[\eta(\mu)])$.
It is easy to prove that ${P}_{\mathrm{av}}$ satisfies the following:
${P}_{\mathrm{av}}=\arg\min_{\rho\in\P(\X)}\E_{\mu\sim{P}}[\W(\mu,\rho)^{2}]$,
which is the definition of the Fr\'echet mean, see for example \cite{OTBook}.
\textit{In this setting, $\M_{\eta\#P}=\M_{\eta\#Q}$ boils down to having
the same Fr\'echet means, ${P}_{\mathrm{av}}={Q}_{\mathrm{av}}$.}\qed
\end{example}

We will later see that the Hilbert embedding corresponding to the
intrinsic sliced 2-Wasserstein distance is assembled of embeddings
like in Example 1 applied after pushforwards (see Figure \ref{fig:pushforward}
for an intuition). This means that the resulting equality $\M_{\eta\#P}=\M_{\eta\#Q}$
becomes more stringent, making it a better proxy for detecting the
deviations from $P=Q$ without losing the interpretability aspect.

\section{\label{sec:Intrinsic-Sliced-2-Wasserstein}Intrinsic Sliced 2-Wasserstein
Distance }

We introduce a Hilbertian version of $\W$ on manifolds and graphs
via a construction we call \emph{intrinsic slicing} due to its use
of the domain's intrinsic geometric properties. To focus our discussion
we concentrate on the manifold case, as the graph case is simpler
and is obtained by replacing the Laplace-Beltrami operator by the
graph Laplacian. 

Let $\lambda_{\ell},\phi_{\ell};\ell=0,1,....$ be the eigenvalues
and eigenfunctions of the Laplace-Beltrami operator on $\X$ with
Neumann boundary conditions. The eigenfunctions are sorted by increasing
eignevalue and assumed to be orthonormal with respect to some fixed
(e.g. uniform) measure on $\X$; also $\phi_{0}=const$ and $\lambda_{0}=0$.
One can define the spectral kernel $k(x,y)=\sum_{\ell}\alpha(\lambda_{\ell})\phi_{\ell}(x)\phi_{\ell}(y)$
and the corresponding spectral distance on the manifold $d(x,y)=k(x,x)+k(y,y)-2k(x,y)=\sum\alpha(\lambda_{\ell})(\phi_{\ell}(x)-\phi_{\ell}(y))^{2}$,
where $\alpha:\mathbb{R}_{\geq0}\rightarrow\mathbb{R}_{\geq0}$ is
a function that controls contribution from each spectral band. By
setting $\alpha(\lambda)=e^{-t\lambda}$ for some $t>0$, we get the
heat/diffusion kernel and the corresponding diffusion distance \cite{diffusion_map}.
Another important case is $\alpha(\lambda)=1/\lambda^{2}$ if $\lambda>0$
and $\alpha(0)=0$, which gives the biharmonic kernel and distance
\cite{Biharmonic}. In both of these constructions $\alpha(\cdot)$
is a decreasing function, allowing the smoother low-frequency (i.e.
smaller $\lambda_{\ell}$) eigenfunctions to contribute more.

\subsection{Definition and properties}
\label{subsec:defn}

A real-valued function $\phi:\X\rightarrow\mathbb{R}$ can be used
to map the manifold $\X$ onto the real line. Any probability measure
$\mu\in\P(\X)$ can likewise be projected onto the real line using
the pushforward of $\phi$, which we denote by $\phi\sharp\mu=\mu\circ\phi^{-1}\in\P(\R)$.
While the pushforward notions used here and in previous sections are
conceptually the same, for clarity we use $\sharp$ for measures and
$\#$ for meta-distributions. We define intrinsic slicing as follows.
\begin{defn}
\label{def: isd}Given a function $\alpha:\mathbb{R}_{\geq0}\rightarrow\mathbb{R}_{\geq0}$
and a probability distance $\D(\cdot,\cdot)$ on $\P(\R)$, we define
the intrinsic sliced distance $IS\D(\cdot,\cdot)$ on $\P(\X)$ by
\[
IS\D^{2}(\mu,\nu)=\sum_{\ell}\alpha(\lambda_{\ell})\D^{2}(\phi_{\ell}\sharp\mu,\phi_{\ell}\sharp\nu).
\]
The choice of the Laplacian eigenfunctions in the definition can be
justified by a number of their properties. Eigenfunctions are intrinsic
quantities of a manifold and are ordered by smoothness. Thus, they
allow capturing the intrinsic connectivity of the underlying domain.
Furthermore, due to the orthogonality of eigenfunctions, their pushforwards
can capture complementary aspects of the distribution. 
\end{defn}
While the definition is general, our focus in this paper is on the
case when $\D=\W$; we remind that we always use $\W$ to denote the
2-Wasserstein distance on $\P(\R)$. We call the resulting distance
\emph{Intrinsic Sliced 2-Wasserstein Distance}, and denote it by $IS\W$.
First, we discuss the convergence of the infinite sum in Definition
\ref{def: isd}. 
\begin{prop}
\label{Prop-well-defined}If $\X$ is a smooth compact $n$-dimensional
manifold and $\sum_{\ell}\lambda_{\ell}^{(n-1)/2}\alpha(\lambda_{\ell})<\infty$,
then $IS\W$ is well-defined. 
\end{prop}

Next, we prove a number of properties of $IS\D$.
\begin{prop}
\label{prop:hilbertian}If $\D$ is a Hilbertian probability distance
such that $IS\D$ is well-defined, then (i) $IS\D$ is Hilbertian,
and (ii) $IS\D$ satisfies the following metric properties: non-negativity,
symmetry, the triangle inequality, and $IS\D(\mu,\mu)=0$.
\end{prop}
\begin{proof}
By Hilbertian property of $\D$, there exists a Hilbert space $\H^{0}$
and a map $\eta^{0}:\P(\R)\rightarrow\H^{0}$ such that $\D(\rho_{1},\rho_{2})=\Vert\eta^{0}(\rho_{1})-\eta^{0}(\rho_{2})\Vert_{\mathcal{\H^{0}}}$
for all $\rho_{1},\rho_{2}\in\P(\R)$. Plugging this into Definition
\ref{def: isd} we have $IS\D(\mu,\nu)=\Vert\eta(\mu)-\eta(\nu)\Vert_{\mathcal{\H}}$,
where $\H=\oplus_{\ell}\H^{0}$ and the $\ell$-th component of $\eta(\mu)$
is $\sqrt{\alpha(\lambda_{\ell})}\eta_{0}(\phi_{\ell}\sharp\mu)\in\H$.
The second part of Proposition \ref{prop:hilbertian} directly follows
from the Hilbert property.
\end{proof}
Since $\W$ is Hilbertian on $\P(\R)$, the application of Proposition
\ref{prop:hilbertian} yields that $IS\W$ is also Hilberitan, making
it possible to use $IS\W$ for our hypothesis tests in Section \ref{sec:Hypothesis-testing}.

% The following result shows that $IS\W$ inherits an important property
% of the Wasserstein distances, namely that the distance between two
% Dirac delta measures equals to a specific ground distance between
% their locations.
% \begin{prop}
% \label{prop:ground-dist}When $\mu=\delta_{x}(\cdot),\nu=\delta_{y}(\cdot)$
% for two points $x,y\in\X$, we have $IS\W(\mu,\nu)=d(x,y)$, where
% $d(\cdot,\cdot)$ is the spectral distance corresponding to the choice
% of $\alpha(\cdot)$.
% \end{prop}

For a simple choice of distance $\D$ on $\P(\R)$, namely the absolute
mean difference, the corresponding intrinsic sliced distance is the
well-known MMD \cite{mmd}. 
\begin{prop}
\label{prop:MMD-equiv}Let $\D(\rho_{1},\rho_{2})=\vert\E_{x\sim\rho_{1}}[x]-\E_{y\sim\rho_{2}}[y]\vert$
for $\rho_{1},\rho_{2}\in\P(\R)$, then the corresponding $IS\D$
 is equivalent to the MMD with the spectral kernel $k(\cdot,\cdot)$.
\end{prop}
When $k(x,y)$ is the heat kernel, the sliced distance in Proposition
\ref{prop:MMD-equiv} is very much like the MMD with the Gaussian
kernel, with the parameter $t$ in $\alpha(\lambda)=e^{-t\lambda}$
controlling the kernel width. Indeed, the two kernels coincide on
$\mathbb{R}^{d}$, and on general manifolds Varadhan's formula gives
asymptotic equivalence for small $t$ \cite{varadhanformula}. 

An interesting insight derived from the above result is that $IS\W$
is in a sense a ``stronger'' distance than MMD that uses the corresponding
spectral kernel. The $IS\W$ compares the quantiles of the pushforward
distributions (Eq. (\ref{eq:w2-via-quantiles})), whereas MMD compares
their expectations only. We formalize this notion in the next result,
also providing a theoretical reason for preferring $IS\W$ for hypothesis
testing. 
\begin{prop}
\label{prop:stronger-than-MMD}$MMD(\mu,\nu)\leq IS\W(\mu,\nu)$ when
the same $\alpha(\cdot)$ is used in both constructions.
\end{prop}
% \begin{proof}
% This follows directly from the fact that for $\rho_{1},\rho_{2}\in\P(\R)$
% the inequality $\vert\E_{x\sim\rho_{1}}[x]-\E_{y\sim\rho_{2}}[y]\vert\leq\W(\rho_{1},\rho_{2})$
% holds.
% \end{proof}
We are now in a position to prove that $IS\W$ is a true metric.

\begin{thm}
\label{thm:metric}If $\alpha(\lambda)>0$ for all $\lambda>0$ ,
then $IS\W$ is a metric on $\P(\X)$.
\end{thm}
We remind that 2-Wasserstein distance can be defined directly on $\P(\X)$
using the geodesic distance as the ground metric; we denote this distance
as $\W^{\X}$. Lipschitz properties of the eigenfunctions imply the
following:
\begin{prop}
\label{prop:lipschitz}There exists a constant $c$ depending only
on $\X\subseteq\R^{n}$ such that for all $\mu,\nu\in\P(\X)$ the
inequality $IS\W(\mu,\nu)\leq c\W^{\X}(\mu,\nu)\sqrt{\sum_{\ell}\lambda_{\ell}^{(n+3)/2}\alpha(\lambda_{\ell})}$
holds.
\end{prop}
Our final result looks at the quantity $\mathbb{T}$ defined using
$IS\W$ by Eq. (\ref{eq:defn}). We will be using $\mathbb{T}$ computed
on finite collections of measures as a test statistic in the next
section. We show that it enjoys robustness with respect to small perturbations
of the measures in the collection.
\begin{prop}
\label{prop:robust}Let $\{\mu_{i}\}_{i=1}^{N}$ and $\{\nu_{i}\}_{i=1}^{N}$
be two collections of probability measures on $\P(\X)$, such that
$\forall i,\W^{\X}(\mu_{i},\nu_{i})\leq\epsilon$, then $\mathbb{T}(\{\mu_{i}\}_{i=1}^{N},\{\nu_{i}\}_{i=1}^{N})\leq C^{2}\epsilon^{2}$.
Here $C=c\sqrt{\sum_{\ell}\lambda_{\ell}^{(n+3)/2}\alpha(\lambda_{\ell})}$
from previous proposition and is assumed to be finite.
\end{prop}
This bound implies that if the distributions in a collection undergo
horizontal shifts that are small as measured by the geodesic distance
$\W^{\X}$ , then $\mathbb{T}$ is small as well.

\subsection{\label{subsec:Approximate-Hilbert-Embedding}Approximate Hilbert
Embedding}

An important aspect of $IS\W$ is that its Hilbert map $\eta:\P(\X)\rightarrow\mathcal{\H}$
can be approximated by a finite-dimensional embedding $\eta_{D}:\P(\X)\rightarrow\R^{D}$
such that $IS\W(\mu,\nu)\approx\Vert\eta_{D}(\mu)-\eta_{D}(\nu)\Vert_{\R^{D}}$.
This is useful for practical computation and for one of our hypothesis
testing approaches. 

Using the formula for $IS\W$ on $\P(\R)$ in terms of the quantile
function, Eq. (\ref{eq:w2-via-quantiles}), the Hilbert map is defined
by $\eta^{0}(\mu)=F_{\mu}^{-1}$. We have $\W(\mu,\nu)=\Vert\eta^{0}(\mu)-\eta^{0}(\nu)\Vert_{L_{2}(\R)}$,
where the norm involves integration. We can discretize the integral
using the Riemann sum for equidistant knots $s_{k}=\frac{k-1}{D'},k=1,...,D'$,
define the approximate embedding $\eta_{D'}^{0}:\P(\R)\rightarrow\R^{D'}$
as:
\begin{equation}
\eta_{D'}^{0}:\mu\rightarrow\frac{1}{\sqrt{D'}}[F_{\mu}^{-1}(s_{1}),...,F_{\mu}^{-1}(s_{D'})].\label{eq:1-dim-embedding}
\end{equation}
Now, $\W(\mu,\nu)\approx\Vert\eta_{D'}^{0}(\mu)-\eta_{D'}^{0}(\nu)\Vert_{\R^{D'}}$
with approximation quality depending on the embedding dimension $D'$. 

To approximate the Hilbert map for $IS\W$ we truncate the series
defining $IS\W$ and use a finite number of eigenfunctions for pushforward:
$\phi_{\ell},\ell=1,...,L$, where $\phi_{0}$ is dropped since it
is a constant. By inspecting the proof of Proposition \ref{prop:hilbertian}
and using Eq. (\ref{eq:1-dim-embedding}), we can define $\eta_{D}:\P(\X)\rightarrow\R^{D}$
with $D=LD'$ as the concatenation of $L$ maps:
\[
(\eta_{D})_{\ell}:\mu\rightarrow\sqrt{\frac{\alpha(\lambda_{\ell})}{D'}}[F_{\phi_{\ell}\sharp\mu}^{-1}(s_{1}),...,F_{\phi_{\ell}\sharp\mu}^{-1}(s_{D'})].
\]
Spectral decompositions of Laplace-Beltrami operators for general
manifolds \cite{diffusion_map,eigencomp} or graph Laplacians can
be computed numerically. For applications involving simple manifolds,
eigenvalues and eigenfunctions can be computed analytically (see Appendix~\ref{subsec:app_compute}).

The hyperparameters $L$ and $D'$ capture distinct aspects of the complexity of the distribution. A larger $L$ allows access to higher order eigenfunctions that possess greater spatial oscillations, linking to the finer details of the geometry and topology of the underlying domain and distribution collection. On the other hand, $D'$ captures how well the pushforward distribution is represented, which is determined by the number of quantiles. Both hyperparameters can have a significant impact on the success of the testing process.

\section{\label{sec:Hypothesis-testing}Hypothesis Testing}

Let $\{\mu_{i}\}_{i=1}^{N_{1}}$ and $\{\nu_{i}\}_{i=1}^{N_{2}}$
be two i.i.d. collections of measures drawn from $P,Q\in\P(\P(\X))$
respectively. Our goal is to use these samples to test the null hypothesis
$H_{0}:\M_{\eta\#P}=\M_{\eta\#Q}$, where $\eta$ is the Hilbert embedding
of the sliced distance $IS\W$ on $\P(\X)$. 

\subsection{\label{subsec:Resampling-Based-Test}Resampling Based Test}

We use the quantity $\mathbb{T}(\cdot,\cdot)$ from Eq. (\ref{eq:defn})
as the test statistic. Its sample version is computed by replacing
the expectations by the empirical means, and excluding the diagonal
terms to achieve unbiasedness
\begin{align}
\hat{\mathbb{T}}\equiv & \sum_{i,j:i\neq j}\frac{IS\W^{2}(\mu_{i},\mu_{j})}{2N_{1}(N_{1}-1)}+\sum_{i,j:i\neq j}\frac{IS\W^{2}(\nu_{i},\nu_{j})}{2N_{2}(N_{2}-1)} \label{eq:T-statistic-isd}\\
& -\sum_{i,j}\frac{IS\W^{2}(\mu_{i},\nu_{j})}{N_{1}N_{2}}. \nonumber
\end{align}
Note that $\E\hat{\mathbb{T}}=\mathbb{T}({P},{Q})$. In practice,
the $IS\W$ values are computed from the approximate embedding: $IS\W(\rho_{1},\rho_{2})\approx\Vert\eta_{D}(\rho_{1})-\eta_{D}(\rho_{2})\Vert_{\R^{D}}$.
We denote the resulting statistic by $\tilde{\mathbb{T}}_{L,D'}$.

The difference between $\tilde{\mathbb{T}}_{L,D'}$ and the population
version (i.e. $\mathbb{T}-\tilde{\mathbb{T}}_{L,D'}$) can be decomposed
as $(\mathbb{T}-\hat{\mathbb{T}})+(\mathbb{\hat{T}}-\hat{\mathbb{T}}_{L})+(\hat{\mathbb{T}}_{L}-\tilde{\mathbb{T}}_{L,D'})$,
where the summands inside the terms $\hat{\mathbb{T}}_{L}$ and $\tilde{\mathbb{T}}_{L,D'}$
correspond to partial sums that approximate $IS\W^{2}(\cdot,\cdot)$
by $\sum_{l=1}^{L}\alpha(\lambda_{l})\W^{2}(\phi_{l}\sharp\cdot,\phi_{l}\sharp\cdot)$,
and $\W^{2}(\phi_{l}\sharp\cdot,\phi_{l}\sharp\cdot)$ by $\|\eta_{D'}(\phi_{l}\sharp\cdot)-\eta_{D'}(\phi_{l}\sharp\cdot)\|^{2}$,
respectively. We show in Appendix~\ref{subsec:app_resamp_proofs} that a) summands in the second and
third terms in the sum can be made infinitesmally small by choosing
large enough $L$ and $D'$, respectively; b) an asymptotic result
for the first difference can be obtained by extending the tools from \citet{mmd,serflingbook}.
These results are based on several assumptions detailed in Appendix~\ref{subsec:app_resamp_proofs}. 
% \begin{description}
% \item [{(i)}] $\sum_{\ell=1}^{\infty}\alpha(\lambda_{\ell})\lambda_{\ell}^{(n+3)/2}<\infty$, 
% \item [{(ii)}] The pushforward meta-distributions $\eta\#P,\eta\#Q$ are
% square-integrable: $\mathbb{E}_{\mu,\mu'\sim P}\langle\eta(\mu),\eta(\mu')\rangle_{\mathcal{H}}<\infty$
% and $\mathbb{E}_{\nu,\nu'\sim Q}\langle\eta(\nu),\eta(\nu')\rangle_{\mathcal{H}}<\infty$,
% \item [{(iii)}] The meta-distributions $P$ and $Q$ generate absolutely
% continuous densities. 
% \end{description}
Combining the two results, we establish asymptotic distributions of
$\tilde{\mathbb{T}}_{L,D'}$:
\begin{thm}
\label{Thm:asy-tilde} Assume relevant conditions (see Appendix~\ref{subsec:app_resamp_proofs}) hold. Define $N=N_{1}+N_{2}$,
and suppose that as $N_{1},N_{2}\rightarrow\infty$, we have $N_{1}/N\rightarrow\rho_{1},N_{2}/N\rightarrow\rho_{2}=1-\rho_{1}$,
for some fixed $0<\rho_{1}<1$.With $L\geq L_{N},D'\geq D_{N}$ chosen
in an appropriate way (see Appendix~\ref{subsec:app_resamp_proofs}), under $H_{0}:\M_{\eta\#P}=\M_{\eta\#Q}$
we have
% $N\tilde{\mathbb{T}}_{L,D'}\leadsto\sum_{m=1}^{\infty}\gamma_{m}(A_{m}^{2}-1)$,
\[
N\tilde{\mathbb{T}}_{L,D'}\leadsto\sum_{m=1}^{\infty}\gamma_{m}(A_{m}^{2}-1),
\]
where $A_{m}\sim N(0,1)$ for $m=1,2,\ldots$, and $\gamma_{m}$ are
the eigenvalues of a certain operator that depends on $P$ and $Q$.
Further, under $H_{1}:\M_{\eta\#P}\neq\M_{\eta\#Q}$ ,$\sqrt{N}\left(\tilde{\mathbb{T}}_{L,D'}-\mathbb{T}\right)$
is asymptotically Gaussian with mean 0 and finite variance.
\end{thm}
We evaluate the power performance of the testing procedure based on
$\tilde{\mathbb{T}}_{L,D'}$ for the sequence of contiguous alternatives
$H_{1N}=\lbrace(P,Q):\M_{\mu\#P}=\M_{\mu\#Q}+\delta_{N},l=1,2,\ldots\rbrace$,
where the deviation from null is quantified collectively by pushforward
differences $\delta_{\ell N}\in\mathcal{H},\delta_{N}=\oplus{}_{\ell}(\sqrt{\alpha_{\ell}}\delta_{\ell N})$
that are made to approach 0 as $N\rightarrow\infty$. The following
theorem establishes consistency of our testing procedure against a
family of such local alternatives.
\begin{thm}
\label{thm:perm-power}Assume the same conditions and choice of $L,D'$
as Theorem~\ref{Thm:asy-tilde}. Then for the sequence
of contiguous alternatives $H_{1N}$ such that $N\|\delta_{N}\|_{\mathcal{H}^{*}}^{2}\rightarrow\infty$,
the test based on $\tilde{\mathbb{T}}_{L,D'}$ is consistent for any
$\alpha\in(0,1)$, that is as $N\rightarrow\infty$ the asymptotic
power approaches 1.
\end{thm}

\paragraph*{Testing Procedure}

In practice, to obtain the $p$-value for the $\tilde{\mathbb{T}}_{L,D'}$-statistic
we use a bootstrap procedure. Remember that $\tilde{\mathbb{T}}_{L,D'}$
is computed via the approximate embedding $\eta_{D}$ with $D=LD'$.
The collection $\{\mu_{i}\}_{i=1}^{N_{1}}$ is mapped to the collection
$\{X_{i}=\eta_{D}(\mu_{i})\}_{i=1}^{N_{1}}$ of vectors in $\R^{D}$
drawn in an i.i.d. manner from $\eta_{D}\#P=P\circ\eta_{D}^{-1}\in\P(\R^{D})$.
Similarly, for the other collection we have a sample $\{Y_{i}=\eta_{D}(\nu_{i})\}_{i=1}^{N_{2}}$
drawn from $\eta_{D}\#Q$. Now, the null $H_{0}:\M_{\eta\#P}=\M_{\eta\#Q}$
implies that the means of the distributions $\eta_{D}\#P$ and $\eta_{D}\#Q$
coincide in $\R^{D}$. 

The bootstrap null distribution for $\tilde{\mathbb{T}}_{L,D'}$ can
be obtained as follows. Let $\bar{X}$ and $\bar{Y}$ be the sample
means; construct the combined sample $\{X_{i}-\bar{X}+\frac{\bar{X}+\bar{Y}}{2}\}_{i=1}^{N_{1}}\bigcup\{Y_{i}-\bar{Y}+\frac{\bar{X}+\bar{Y}}{2}\}_{i=1}^{N_{2}}$.
This centers both samples at $\frac{\bar{X}+\bar{Y}}{2}$. Now, from
the combined sample we select with replacement $N_{1}$ (resp. $N_{2}$)
samples to make bootstrap sample $\{X_{i}^{b}\}_{i=1}^{N_{1}}$ (resp.
$\{Y_{i}^{b}\}_{i=1}^{N_{2}}$). Repeat this process $B$ times (we
take $B=1000$ in our experiments), and collect the null test statistic
values $\tilde{\mathbb{T}}_{L,D'}^{b}=\tilde{\mathbb{T}}_{L,D'}(\{X_{i}^{b}\}_{i=1}^{N_{1}},\{Y_{i}^{b}\}_{i=1}^{N_{2}})$
for $b=1,...,B$. The approximate $p$-value is then given by: $p=\frac{1}{B+1}\left(\vert\{b:\mathbb{\tilde{T}}_{L,D'}^{b}\geq\tilde{\mathbb{T}}_{L,D'}\}\vert+1\right)$.

\begin{rem}
Permutation testing cannot be applied here as it would detect
differences beyond the mean inequality. Such differences help reject
the stronger null hypothesis $H_{0}:P=Q$ which can be the intent
in some situations. However, such a rejection will not allow pinpointing the aspect responsible
for the difference; see \cite{permuteornot}.
%In a similar context, for hypothesis testing on stochastic processes,
%\cite{kernel-indep} previously showed the inability of the permutation
%procedure to maintain nominal size in the context of independence
%testing.
\end{rem}

\subsection{\label{subsec:Combination-Approach}Testing via $p$-value Combination}

The bootstrap test above incurs a high computational cost and the
granularity of the $p$-values is determined by the number of resamples,
which can be too coarse in massive multiple comparison settings often
seen in industrial applications. Thus, we propose an approach that
avoids resampling. 

As explained above, testing $H_{0}:\M_{\eta\#P}=\M_{\eta\#Q}$ can
be interpreted as testing whether the means of the distributions $\eta_{D}\#P$
and $\eta_{D}\#Q$ coincide in $\R^{D}$. To this end, we adopt the
approach proposed by \citet{AKME} in a spatial statistics context.

First, we apply the Behrens-Fisher-Welch $t$-test (without assuming
equality of variances) to each coordinate of the samples $\{X_{i}=\eta_{D}(\mu_{i})\}_{i=1}^{N_{1}}$
and $\{Y_{i}=\eta_{D}(\nu_{i})\}_{i=1}^{N_{2}}$ to obtain the $p$-values
$p_{k},k=1,2,...,D$. Second, an overall $p$-value is computed via
the harmonic mean $p$-value combination method which is robust to
dependencies \cite{HarmonicP1958,HarmonicP}: 
$$p^{H}=H\left(D/(\frac{1}{p_{1}}+\frac{1}{p_{2}}+\cdots+\frac{1}{p_{D}})\right),$$
where the function $H$ has a known form described in \cite{HarmonicP}.

Another approach for combining $p$-values is the Cauchy combination
test \cite{CauchyP}, but in our numerical experiments we found that
the Cauchy combination approach encounters problems when any of the
$p$-values is very close to $1$, which can happen in our setting
due to the form of the embedding $\eta_{D}$. Therefore, in contrast
to \citet{AKME}, for us the harmonic combination
is the only appropriate choice.

To guarantee size control, we establish a version of Theorem 1 from
\citet{CauchyP} for the harmonic mean $p$-value. Assume that a test
statistic $Z\in\R^{D}$ has null distribution with zero mean and every
pair of coordinates of $Z$ follows bivariate Gaussian distribution.
Compute the coordinate-wise two-sided $p$-values $p_{k}=2(1-\Phi(\vert Z_{k}\vert))$
where $\Phi$ is the standard Gaussian CDF. 
\begin{thm}
Let $p_{k},k=1,...,D$ be the null $p$-values as above and $p^{H}$
computed via harmonic mean approach, then $\lim_{\alpha\to0} \mathrm{Prob}\{p^{H}\leq\alpha\} / \alpha = 1$.
% \[
% \lim_{\alpha\to0}\frac{\mathrm{Prob}\{p^{H}\leq\alpha\}}{\alpha}=1.
% \]
\end{thm}
% In Appendix~\ref{subsec:app_comb_proofs} we show that this theorem applies in our setting, so t
The proposed procedure asymptotically controls the size of the
test for small $\alpha$ (Appendix~\ref{subsec:app_comb_proofs}). Our experimental results show that the control
is already achieved for moderate sample sizes and the commonly used
$\alpha=0.05$. 

\begin{figure}
\noindent\begin{minipage}{1\columnwidth}%
\begin{center}
\begin{tabular}{cc}
\includegraphics[width=0.45\columnwidth]{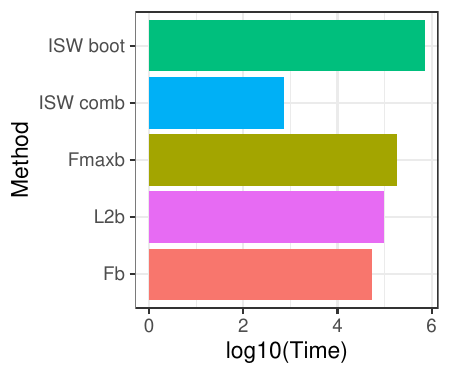} & \includegraphics[width=0.45\columnwidth]{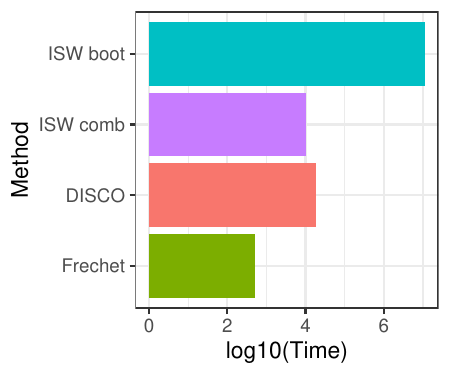}\tabularnewline
(a) & (b) \tabularnewline
\end{tabular}
\par\end{center}
\vspace{-4pt}
\caption{\label{fig:times}Running times for (a) real line and (b) circle experiments.}
\end{minipage}
\end{figure}

\begin{figure*}
\begin{centering}
\begin{tabular}{cccc}
\includegraphics[width=0.23\textwidth]{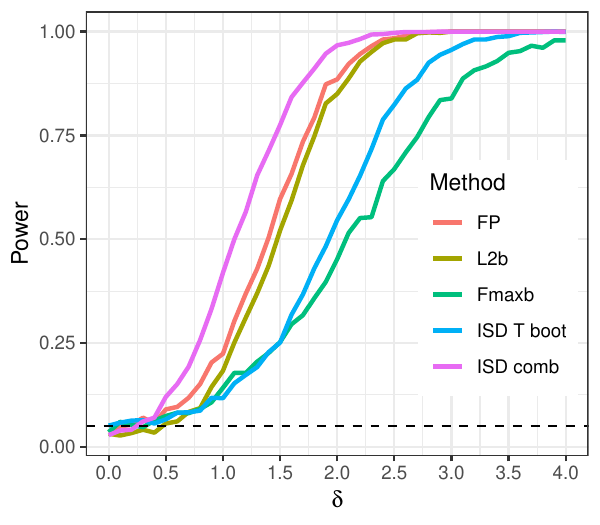} & \includegraphics[width=0.23\textwidth]{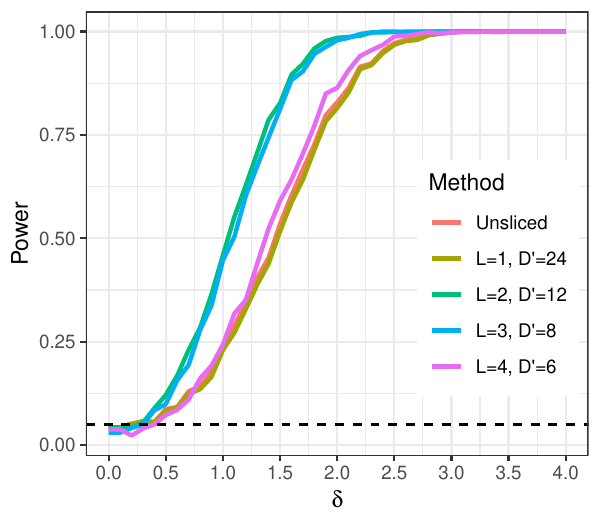} & \includegraphics[width=0.23\textwidth]{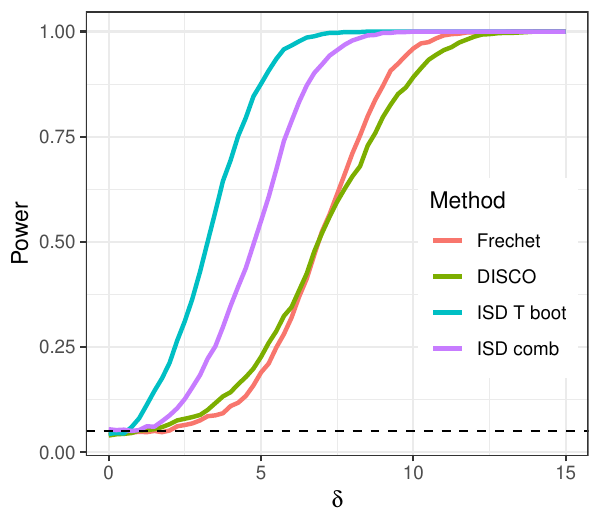} & \includegraphics[width=0.23\textwidth]{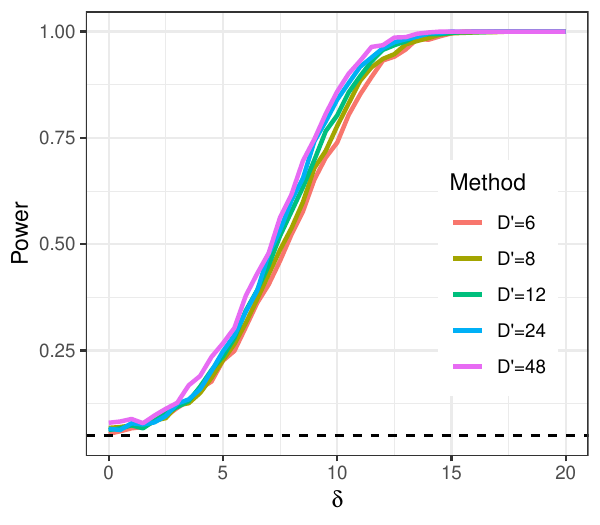}\tabularnewline
(a) & (b) & (c) & (d)\tabularnewline
\end{tabular}
\par\end{centering}
\caption{
\label{fig:power}
Performances on synthetic data. Dotted lines indicates nominal size of all tests ($\alpha=0.05$). (a) comparison with existing methods on finite interval---a test based on basis function representation \citep[FP]{FPtest}, a sum-type $\ell_{2}$ norm-based test (L2b) \cite{L2btest}, and a max-type test \cite{Fmaxbtest} that uses the maximum of coordinate-wise $F$ statistic (Fmaxb), (b) unsliced vs. different settings of $(L,D')$ on finite interval. (c) Circular data, comparing with Fr\'echet ANOVA \cite{10.1093/biomet/asz052}, and the DISCO nonparametric test \cite{disco}; (d) harmonic combination tests on cylindrical data for $L=4$. .}
\end{figure*}
\section{Experiments\label{subsec:Simulations}}

We compare the performance of our tests
with several existing methods, across synthetic and real data, and settings of the embedding parameters $L,D'$. For evaluation, we use empirical power at different degrees of departure from the null hypothesis (Captured by $\delta$ in the plots in Fig.~\ref{fig:power}), calculated by averaging the proportion of rejections at level $\alpha=0.05$ over 1000 independent datasets, with samples divided into two groups of sizes $n_{1}=60,n_{2}=40$. To ensure the tests are well-calibrated, we calculate nominal sizes assuming the two sample groups are drawn from the same meta-distribution. Details on the experiment  settings and computational complexity are in Appendix~\ref{subsec:app_simulations}.

\paragraph*{Finite intervals}

% To obtain our base measures $\mu_{i},\nu_{i},$ we generate bin probabilities
% as (shifted and normalized) values of the function $f(t_{j})=\mu(t_{j})+\alpha(t_{j})$
% at $m=30$ fixed design points $t_{j}=j/(m+1),j=\{1,2,\ldots,m\}$,
% and
% \begin{eqnarray*}
% \mu(t_{j}) & = & 1.2+2.3\cos(2\pi t_{j})+4.2\sin(2\pi t_{j}),\\
% \alpha(t_{j}) & = & \epsilon_{0}+\sqrt{2}\epsilon_{1}\cos(2\pi t_{j})+\sqrt{3}\epsilon_{2}\sin(2\pi t_{j}),
% \end{eqnarray*}
% where $\epsilon_{0,}\epsilon_{1},\epsilon_{2}\sim N(0,1)$ clipped
% between $[-3,3]$. Group 1 and 2 samples are obtained as $\mu_{i}(\cdot)\equiv f(\cdot)$
% and $\nu_{i}(\cdot)\equiv f(\cdot)+\delta$ respectively, where $\delta\in[0,4]$
% is a constant. To make the sample functions non-negative, we shift
% all functions by $M=3(1+\sqrt{2}+\sqrt{3})$. Finally, as the $m$-length
% vector of bin counts for a sample, we generate a random vector from
% the Multinomial distribution with 1000 trials, $m$ outcomes and the
% outcome probabilities proportional to the shifted functional observations
% corresponding to that sample. 

We use embedding dimensions $L=3,D'=10$ to compare our method against
11 functional ANOVA tests---we report results for 3 of
them in Figure~\ref{fig:power}a (see Appendix~\ref{sec:appB} for
complete results). All methods maintain nominal size for $\delta=0$
(Figure \ref{fig:power}a). While the combination test (ISD comb)
based on our proposal outperformed all the other tests across all
values of $\delta$, the bootstrap test that uses the overall $\mathbb{T}$
statistic (ISD T boot) performs better than Fmaxb but worse than others. In terms of computation time, ISD comb takes the least amount of time (Fig.~\ref{fig:times}a).

We also compare the $p$-value combination test based on an \emph{unsliced}
24-dimensional inverse CDF embedding with sliced $IS\W$-based tests
(Figure \ref{fig:power}b). We use multiple pairs of $(L,D')$
values, all of them giving overall embeddings of dimension $D=LD'=24$.
The performance of an $IS\W$-based test that uses slicing over only
the first eigenfunction is almost as good as the unsliced version.
With more eigenfunctions, the powers first improve considerably, then
become similar to the unsliced version again.

\paragraph*{Manifold domains}

We consider data from distributions on circles and cylinders. For
circular data, we take von Mises distributions with randomly chosen
parameters as our samples.
For an angle $x$ (measured in radians),
the von Mises probability density function is given by $f(x|\mu,\kappa)=\exp[\kappa\cos(x-\mu)](2\pi I_{0}(\kappa))^{-1}$,
where $I_{0}(\kappa)$ is the modified Bessel function of order 0.
We fix $\kappa=2$, and use $\mu\equiv\mu_{i}\sim N(0,0.1^{2})$,
$\mu\equiv\nu_{i}\sim N(\delta,0.1^{2})$ for samples from group 1
and 2 respectively---with $\delta\in[0,15]\times\pi/180$ (i.e. 0
to 15 degrees converted to radians). As each observation vector, we
take 100 random draws from each sample-specific distribution.
For our embeddings, we use $L=10,D'=20$. Since the competing methods
cannot handle circular geometry directly, to implement them we cut
the circle into an interval. Figure \ref{fig:power}c
%\ref{fig:power-manifold}a 
shows that all methods maintain nominal size, but both our tests maintain
considerably higher power than existing methods for all $\delta$. Computationally, ISD comb has comparable order of magnitude as the other two methods (Fig.~\ref{fig:times}b).

To generate cylindrical data, we use the distribution proposed by
\citet{MardiaSutton}.
Samples from each distribution have the form
of a bivariate random vector $(\Theta,X)$, where the first (circular)
marginal $\Theta$ has a von Mises distribution, and $X$ is a Gaussian
conditional on $\Theta=\theta$. We draw the
mean parameters for each coordinate-wise distribution
from $\text{{Unif}}(0,1)$ and $\text{{Unif}}(\delta,\delta+1)$ for
sample groups 1 and 2 respectively, with $\delta\in[0,30]\times\pi/180$.
We then use 500 random draws from each sample distribution to obtain
histograms.
To evaluate the effects of choosing $L,D'$ we calculate
our embeddings for $L\in\lbrace2,3,4,5\rbrace,D'\in\lbrace6,8,12,24,48\rbrace$.
The choice of $L$ has small effect on performance, so we report
results for $L=4$ in Figure \ref{fig:power}d. Higher values of $D'$ result in some increase in power.

\paragraph*{Comparison to existing slicing methods}
% \begin{table}[h]
% \centering
% \begin{tabular}{cccc}
% \toprule
% Method	& Power	& Size &  Time (sec) \tabularnewline
% \midrule
% $IS\W$	& 1.00	& 0.029	& 630.67\tabularnewline
% SW	& 0.128	& 0.029	& 579.3 \tabularnewline
% GSW-circle & 	0.128	& 0.029	& 564.82\tabularnewline
% GSW-poly3	& 0.025	& 0.025	& 779.69\tabularnewline
% GSW-poly5	& 0.044	& 0.035	& 1009.85\tabularnewline
% Sobolev	& 1.00	& 0.048	& 21474.88\tabularnewline
% \bottomrule
% \end{tabular}
% \caption{
% \label{table:comparegraph}
% Comparison of multiple slicing techniques.}
% \end{table}

As an example showcasing the importance of using intrinsic distances, consider the following graph setup: points A and B are adjacent on a planar regular 200-gon, with the edge AB removed. The intrinsic distance between A and B is large and the (extrinsic) Euclidean distance is small. No matter what Euclidean projection we use, the SW cost of moving probability mass from A to B would be small, leading to testing power loss when we compare two distribution sets that concentrate one around A and the other around B.  To numerically analyze this case we generate distributions on this graph by putting Unif(0,1) weights at each of the vertices, and added an additional weight of 10 to the bin at point A to obtain the first set of distributions with a mode at A. For the second set we put the weight of 10 at B instead to obtain a set of distributions with mode at B.

\begin{table}[h]
\centering %
\scalebox{.975}{
\begin{tabular}{cccc}
\toprule 
{\footnotesize{}Method } & {\footnotesize{}Power } & {\footnotesize{}Size } & {\footnotesize{}Time (sec) }\tabularnewline
\midrule 
{\footnotesize{}$IS\W$ } & {\footnotesize{}1.00 } & {\footnotesize{}0.029 } & {\footnotesize{}630.67}\tabularnewline
{\footnotesize{}SW } & {\footnotesize{}0.128 } & {\footnotesize{}0.029 } & {\footnotesize{}579.3 }\tabularnewline
{\footnotesize{}GSW-circle } & {\footnotesize{}0.128 } & {\footnotesize{}0.029 } & {\footnotesize{}564.82}\tabularnewline
{\footnotesize{}GSW-poly3 } & {\footnotesize{}0.025 } & {\footnotesize{}0.025 } & {\footnotesize{}779.69}\tabularnewline
{\footnotesize{}GSW-poly5 } & {\footnotesize{}0.044 } & {\footnotesize{}0.035 } & {\footnotesize{}1009.85}\tabularnewline
{\footnotesize{} MSW } & {\footnotesize{}1.00 } & {\footnotesize{}0.63 } & {\footnotesize{} 61764.22 }\tabularnewline
{\footnotesize{}ST} & {\footnotesize{}1.00 } & {\footnotesize{}0.048 } & {\footnotesize{}21474.88}\tabularnewline
\bottomrule
\end{tabular}
}
\caption{ \label{table:comparegraph} Comparison of multiple slicing techniques.}
\end{table}

Table~\ref{table:comparegraph} shows comparison of $IS\W$ with Sliced Wasserstein (SW) in two dimensions, GSW with 3 choices of defining functions (circular, homogeneous polynomial of degree 3, and degree 5), MSW, and Sobolev transport (ST). We use $L = 10, D' = 8$ to produce all embeddings, and implement the sliced version of ST by aggregating ST distances originating from 10 random vertices. $IS\W$ achieves the highest possible power, maintains nominal type-I error, and has computational time comparable to SW/GSW. ST and MSW achieve the same power as $IS\W$, but take much longer. This is because unlike slicing-based methods, they need to rely on permutation testing in absence of embeddings. For MSW, the high power comes at the price of a size that is much higher than the acceptable threshold of 0.05.

% \section{\label{subsec:Real-Data}Real Data}

\setlength{\tabcolsep}{3pt}

\begin{table*}[t]
\begin{minipage}[c]{0.65\textwidth}%
\begin{flushleft}
{\footnotesize{}}%
\begin{tabular}{ccccccccc}
\toprule 
{\footnotesize{}Age Groups} & {\footnotesize{}6--15} & {\footnotesize{}16--25} & {\footnotesize{}26--35} & {\footnotesize{}36--45} & {\footnotesize{}46--55} & {\footnotesize{}56--65} & {\footnotesize{}66--75} & {\footnotesize{}76--85}\tabularnewline
\midrule 
{\footnotesize{}6--15} &  & {\footnotesize{}0.394} & {\footnotesize{}0.098} & {\footnotesize{}0.555} & {\footnotesize{}0.882} & {\footnotesize{}0.985} & {\footnotesize{}0.919} & {\footnotesize{}0.997}\tabularnewline
{\footnotesize{}16--25} & \textbf{\footnotesize{}1.2e-13} &  & {\footnotesize{}0.575} & {\footnotesize{}0.967} & {\footnotesize{}0.126} & {\footnotesize{}0.921} & {\footnotesize{}0.911} & {\footnotesize{}0.977}\tabularnewline
{\footnotesize{}26--35} & \textbf{\footnotesize{}3.1e-21} & \textbf{\footnotesize{}2.7e-04} &  & {\footnotesize{}0.459} & {\footnotesize{}0.197} & {\footnotesize{}0.996} & {\footnotesize{}0.919} & {\footnotesize{}0.565}\tabularnewline
{\footnotesize{}36--45} & \textbf{\footnotesize{}6.1e-22} & \textbf{\footnotesize{}7.9e-08} & \textbf{\footnotesize{}0.042} &  & {\footnotesize{}0.864} & {\footnotesize{}0.637} & {\footnotesize{}0.849} & {\footnotesize{}0.991}\tabularnewline
{\footnotesize{}46--55} & \textbf{\footnotesize{}8.2e-22} & \textbf{\footnotesize{}4.7e-05} & \textbf{\footnotesize{}0.011} & {\footnotesize{}0.343} &  & {\footnotesize{}0.841} & {\footnotesize{}0.165} & {\footnotesize{}0.554}\tabularnewline
{\footnotesize{}56--65} & \textbf{\footnotesize{}1.3e-25} & \textbf{\footnotesize{}0.001} & \textbf{\footnotesize{}0.001} & \textbf{\footnotesize{}5.6e-05} & \textbf{\footnotesize{}0.003} &  & {\footnotesize{}0.991} & {\footnotesize{}0.962}\tabularnewline
{\footnotesize{}66--75} & \textbf{\footnotesize{}3.6e-35} & \textbf{\footnotesize{}7.8e-12} & \textbf{\footnotesize{}1.5e-11} & \textbf{\footnotesize{}4.6e-15} & \textbf{\footnotesize{}1.8e-13} & \textbf{\footnotesize{}0.001} &  & {\footnotesize{}0.989}\tabularnewline
{\footnotesize{}76--85} & \textbf{\footnotesize{}3.8e-46} & \textbf{\footnotesize{}1.4e-26} & \textbf{\footnotesize{}1.7e-30} & \textbf{\footnotesize{}8.4e-37} & \textbf{\footnotesize{}2.1e-35} & \textbf{\footnotesize{}1.3e-17} & \textbf{\footnotesize{}6.5e-09} & \tabularnewline
\bottomrule
\end{tabular}\caption{\label{tab:nhanes-age-groups} Activity intensity comparison across
age groups in the NHANES data. Below diagonal: $p$-values for
the actual data comparisons. Above diagonal: null $p$-values obtained
by combining and randomly splitting the two groups. Bold entries correspond
to rejected hypotheses with the BH procedure at FDR level 0.1.}
\par\end{flushleft}%
\end{minipage}\hfill{}%
\begin{minipage}[c]{0.32\textwidth}%
\begin{flushleft}
{\footnotesize{}}%
\begin{tabular}{lcc}
\toprule 
\multirow{1}{*}{{\footnotesize{}Crime Type}} & {\footnotesize{}Tue vs Thu} & {\footnotesize{}Tue vs Sat}\tabularnewline
\midrule 
{\footnotesize{}Theft} & {\footnotesize{}0.428} & \textbf{\footnotesize{}4.2e-06}\tabularnewline
{\footnotesize{}Deceptive Pract.} & {\footnotesize{}0.313} & \textbf{\footnotesize{}0.001}\tabularnewline
{\footnotesize{}Battery} & {\footnotesize{}0.430} & \textbf{\footnotesize{}0.001}\tabularnewline
{\footnotesize{}Robbery} & {\footnotesize{}0.119} & \textbf{\footnotesize{}0.003}\tabularnewline
{\footnotesize{}Narcotics} & {\footnotesize{}0.854} & \textbf{\footnotesize{}0.004}\tabularnewline
{\footnotesize{}Criminal Dam.} & {\footnotesize{}0.855} & \textbf{\footnotesize{}0.02}\tabularnewline
{\footnotesize{}Other Offense} & {\footnotesize{}0.931} & {\footnotesize{}0.052}\tabularnewline
{\footnotesize{}Burglary} & {\footnotesize{}0.142} & {\footnotesize{}0.261}\tabularnewline
{\footnotesize{}Assault} & {\footnotesize{}0.997} & {\footnotesize{}0.38}\tabularnewline
{\footnotesize{}Motor Veh. Theft} & {\footnotesize{}0.858} & {\footnotesize{}0.416}\tabularnewline
\bottomrule
\end{tabular}\caption{\label{tab:chicago-crime}Chicago Crime analysis $p$-values. Bold
entries correspond to rejected hypotheses with the BH procedure at
FDR level 0.1.}
\par\end{flushleft}%
\end{minipage}
\end{table*}

\section{Real Data Examples}

\paragraph*{NHANES}

This dataset \cite{nhanes} contains physical activity pattern readings for 6839 individuals, corresponding to activity monitor intensity
values for $24\times60=1440$ minutes throughout the day, over 7 days.
We capture this activity pattern into a cylindrical histogram with
time and intensity dimensions, having 96 and 100
bins respectively. Since the time dimension is periodic, normalized
counts of this histogram can be considered as person-specific probability
distributions over the cylinder $S^{1}(T_{1})\times[0,T_{2})$, with
$T_{1}=96,T_{2}=100$. To check if activity patterns vary across different groups of individuals,
we first split individuals into age-specific groups: 6--15, 16--25,
...,76--85, then sample 100 males and 100 females from each split.
For our analysis, we consider $L=3$ indices along the two directions,
i.e. $\ell_{1},\ell_{2}=1,2,3$ and $D'=5$. We summarize the $p$-value
combination test results in Table \ref{tab:nhanes-age-groups}, \emph{below
the diagonal}. We control false discovery rate at $0.1$ by running the resulting $p$-values through the procedure of \citet{FDR}.
$IS\W$ detects statistically significant differences between all
pairs of groups, except the 36--45 and 46--55 groups. As expected,
the control $p$-values---obtained by mixing male and female samples
in each age group and splitting arbitrarily---do not concentrate
near zero.
% This data \cite{nhanes} contains physical activity pattern readings
% for 6839 individuals. Data for each individual corresponds to activity
% monitor intensity values for 7 days. Since the time dimension is periodic,
% we get person-specific
% probability distributions over the cylinder $S^{1}(T_{1})\times[0,T_{2})$.
% We check if activity patterns vary across age groups. The $p$-value combination test results are shown \emph{below
% the diagonal} in Table \ref{tab:nhanes-age-groups}. Our method detects statistically significant differences
% between all pairs of groups, except the 36--45 and 46--55 groups.
% As expected, the\emph{ control} $p$-values---obtained by mixing
% samples between two age groups and splitting arbitrarily---do not
% concentrate near zero.
More details are given in Appendix~\ref{subsec:app_nhanes}.

\paragraph*{Chicago Crime}
We use this dataset \cite{chicagocrime} to show a practical usage of our method on histograms over graphs. Each beat (geographic
area subdivision used by police) corresponds to a vertex, and two
vertices are connected by an edge if the corresponding beats share
a geographic boundary. For each crime type and day, the normalized
counts of that crime type for each beat gives a daily probability
distribution over the graph. Our goal is to compare the collection of
distributions of, say, theft occurring on Tuesday to those of Thursday
and Saturday. The Tuesday versus Thursday comparison is intended as
a null case, as we do not expect to see any differences between them
\cite{AKME}. Table~\ref{tab:chicago-crime} results $IS\W$ outcomes using 100-dimensional embeddings ($L=20, D'=5$).We detect statistically significant differences between Tuesday and Saturday patterns for six categories of crime, and as
expected, no differences between Tuesday and Thursday patterns.
For more details and another graph-based application, see Appendices~\ref{subsec:app_chic} and \ref{subsec:app_brain}, respectively.

\section{\label{sec:Conclusion}Conclusion}

In this paper we have introduced a novel class of sliced Wasserstein distances on manifolds and graphs, and applied the resulting $IS\W$ distance in the context of hypothesis tests for meta-distributions on general domains, due to the high relevance of this setup in real data situations. It is worth pointing out that the assumption-lean nature of $IS\W$ can enable its adoption by many other statistical and ML methods. 

The $IS\W$ embeddings can be used to do other hypothesis tests, such as tests for equality of covariances by extending the notion of Wasserstein covariances from \cite{wass_covariance} to general domains. The embeddings can also be used as input features for supervised learning problems over distributions.
%where prediction targets live in a general domain. %Given that rigorous
%Fr\'echet mean-based
%methodology for such problems has only been proposed recently \cite{HanEtal19}, development of prediction models for manifold-valued data free of restrictive assumptions is an attractive future line of research. 
In another direction, $IS\W$ is directly applicable as a loss function for generative modeling. For example, if one is generating distributions over a graph or manifold, $IS\W$ can serve as a loss function similarly to how other sliced distances are used in this context over Euclidean domains.

The proposed framework has a few limitations. There is scope of exploration for choosing the parameters $L$ and $D'$ in a principled manner. Empirical computation of eigenfunctions for general manifolds will introduce approximation errors that need to be tackled by expanding our theoretical results. 
%Finally, per Theorem~\ref{thm:metric}, $IS\W$ is theoretically a true metric. But in practice, the reflexiveness property (i.e. $d(x,y) = 0 \Leftrightarrow x = y$) is lost when a finite number of slices have to be used. This is also the case for all SW distances. In spite of that they remain highly effective practically. As our experiments demonstrate, $IS\W$ substantially improves upon the performance of SW distances.
In terms of potential societal impacts of our proposal, any difference between data distributions from different demographic groups found using our method should be evaluated in light of potential biases in the data collection stage. %Releasing the embedding vectors can be risky in privacy-sensitive situations---such as analyzing manifold-valued personal data. Additional studies should quantify such risks and generate differentially private embeddings. 

\bibliographystyle{icml2023}
\bibliography{icml_new}
% \bibliography{biblio}

% \setcounter{prop}{0}
\setcounter{thm}{0}

\newpage
\onecolumn
\section*{Appendix}
\input{supplement}

\end{document}

%% file: supplement.tex
\appendix
For the appendix to be self-contained, we restate results (theory and experiments) in the main paper when necessary. Appendix~\ref{sec:appA} contains proofs of theoretical results and some implementation details. Details and results of experiments performed are in Appendix~\ref{sec:appB}. Code and data behind these experiments are at \url{https://github.com/shubhobm/isw}.

\section{Proofs and additional results}
\label{sec:appA}

\subsection{Proofs and Notes for Section~\ref{sec:prelims}}

\theoremstyle{plain}
\newtheorem*{prop34}{Proposition 3.4}
\begin{prop34}
\label{prop-defn-1}For $P,Q\in\P(\P(\X))$, the following equality
holds:
\begin{equation}
\mathbb{T}({P},{Q})=\E_{\mu\sim{P},\nu\sim{Q}}[\D^{2}(\mu,\nu)]-\frac{1}{2}\E_{\mu,\mu'\sim{P}}[\D^{2}(\mu,\mu')]-\frac{1}{2}\E_{\nu,\nu'\sim{Q}}[\D^{2}(\nu,\nu')],\label{eq:defn-1}
\end{equation}
where to avoid notational clutter we use $\mathcal{D}^{2}(\cdot,\cdot)$
as a shorthand for $(\mathcal{D}(\cdot,\cdot))^{2}$.
\end{prop34}
\begin{proof}
This is a straightforward application of the ``kernel trick'': using
the Hilbert property of the distance we can rewrite,
\begin{align*}
\E_{\mu\sim{P},\nu\sim{Q}} & [\Vert\eta(\mu)-\eta(\nu)\Vert_{\mathcal{\H}}^{2}]-\frac{1}{2}\E_{\mu,\mu'\sim{P}}[\Vert\eta(\mu)-\eta(\mu')\Vert_{\mathcal{\H}}^{2}]-\frac{1}{2}\E_{\nu,\nu'\sim{Q}}[\Vert\eta(\nu)-\eta(\nu')\Vert_{\mathcal{\H}}^{2}]\\
= & \E_{\mu\sim{P}}[\Vert\eta(\mu)\Vert_{\H}^{2}]+\E_{\nu\sim{Q}}[\Vert\eta(\nu)\Vert_{\mathcal{\H}}^{2}]-2\langle\E_{\mu\sim{P}}[\eta(\mu)],\E_{\nu\sim{Q}}[\eta(\nu)]\rangle_{\H}\\
- & \E_{\mu\sim{P}}[\Vert\eta(\mu)\Vert_{\H}^{2}]-\E_{\nu\sim{Q}}[\Vert\eta(\nu)\Vert_{\mathcal{\H}}^{2}]\\
+ & \langle\E_{\mu\sim{P}}[\eta(\mu)],\E_{\mu\sim{P}}[\eta(\mu)]\rangle_{\H}+\langle\E_{\nu\sim{Q}}[\eta(\nu)],\E_{\nu\sim{Q}}[\eta(\nu)]\rangle_{\H}\\
= & \Vert\E_{\mu\sim{P}}[\eta(\mu)]-\E_{\nu\sim{Q}}[\eta(\nu)]\Vert_{\H}^{2}=\mathbb{T}({P},{Q}).
\end{align*}
Which gives the sought equivalence.
\end{proof}

%%%%%%%%%%%%%%%%%%%%%%%%%%%%%%%%%%%%%%%%%%%%%%%%%%%%%%%%%%%%

\subsection{Proofs and Notes for Section~\ref{subsec:defn}}

\theoremstyle{plain}
\newtheorem*{prop42}{Proposition 4.2}
\begin{prop42}
\label{Prop-well-defined-1}If $\X$ is a smooth compact $n$-dimensional
manifold and $\sum_{\ell}\lambda_{\ell}^{(n-1)/2}\alpha(\lambda_{\ell})<\infty$,
then $IS\W$ is well-defined. 
\end{prop42}
\begin{proof}
We use H\"ormander's bound on the supremum norm of the eigenfunctions:
\[
\Vert\phi_{\ell}\Vert_{\infty}\leq c\lambda_{\ell}^{(n-1)/4}\Vert\phi_{\ell}\Vert_{2},
\]
for some constant $c$ that depends on the manifold. By orthonormality
of the eigenfunctions we have $\forall\ell,\Vert\phi_{\ell}\Vert_{2}=1$.
Next, note that $\W(\phi_{\ell}\sharp\mu,\phi_{\ell}\sharp\nu)\leq2\Vert\phi_{\ell}\Vert_{\infty}$
as the maximum distance that the mass would be transported in any
transportation plan involving pushforwards via $\phi_{\ell}$ is upper
bounded by $2\Vert\phi_{\ell}\Vert_{\infty}$. As a result, every
term in the series defining $IS\W$ can be upper-bounded by the terms
of the following series:
\[
\sum_{\ell}4\Vert\phi_{\ell}\Vert_{\infty}^{2}\alpha(\lambda_{\ell})\leq\sum_{\ell}4c^{2}\lambda_{\ell}^{(n-1)/2}\alpha(\lambda_{\ell})\propto\sum_{\ell}\lambda_{\ell}^{(n-1)/2}\alpha(\lambda_{\ell}),
\]
which proves the claim by the direct comparison test for convergence
of series.
\end{proof}
\begin{rem}
When Weyl law applies, we have that $\lambda_{\ell}=\Theta(\ell^{2/n})$,
which allows us to replace the above condition by $\sum_{\ell}\ell^{(n-1)/n}\alpha(\lambda_{\ell})<\infty$.
For the diffusion kernel/distance choice of $\alpha(\lambda)=e^{-t\lambda}$
the series always converges independently of the manifold dimension.
For biharmonic choice of $\alpha(\lambda)=1/\lambda^{2}$, the sufficient
condition is the convergence of $\sum_{\ell}\ell^{(n-1)/n}/\lambda_{\ell}^{2}\sim\sum_{\ell}\ell^{(n-1)/n}/(\ell^{2/n})^{2}=\sum_{\ell}\ell^{(n-5)/n}$,
where we applied Weyl's asymptotic again. As a result, the biharmonic
choice of $\alpha$ is guaranteed to provide a well-defined $IS\W$
for 1 and 2-dimensional manifolds. Notice, however, that the H\"ormander's
bound used in the proof of the above proposition can be rather lax
in some of the settings that are practically relevant, such as the
product spaces of lines and circles (where all of the eigenfunctions
are bounded by a constant as can be seen from Table \ref{tab:Eigenvalues-and-eigenfunctions}),
and, thus, convergence for the biharmonic choice holds more widely.
\end{rem}

\theoremstyle{plain}
\newtheorem*{prop43}{Proposition 4.3}
\begin{prop43}
\label{prop:hilbertian-1}
If $\D$ is a Hilbertian probability distance
such that $IS\D$ is well-defined, then

(i) $IS\D$ is Hilbertian, and 

(ii) $IS\D$ satisfies the following metric properties: non-negativity,
symmetry, the triangle inequality, and $IS\D(\mu,\mu)=0$.
\end{prop43}
\begin{proof}
By Hilbertian property of $\D$, there exists a Hilbert space $\H^{0}$
and a map $\eta^{0}:\P(\R)\rightarrow\H^{0}$ such that $\D(\rho_{1},\rho_{2})=\Vert\eta^{0}(\rho_{1})-\eta^{0}(\rho_{2})\Vert_{\H^0}$
for all $\rho_{1},\rho_{2}\in\P(\R)$. Plugging this into the definition
of $IS\D$ we have $IS\D(\mu,\nu)=\Vert\eta(\mu)-\eta(\nu)\Vert_{\mathcal{\H}}$,
where $\H=\oplus_{\ell}\H^{0}$ and the $\ell$-th component of $\eta(\mu)$
is $\sqrt{\alpha(\lambda_{\ell})}\eta_{0}(\phi_{\ell}\sharp\mu)\in\H^{0}$.
The second part of Proposition \ref{prop:hilbertian-1} directly follows
from the Hilbert property.
\end{proof}

\setcounter{thm}{0}
\begin{prop}
\label{prop:ground-dist-1}
When $\mu=\delta_{x}(\cdot),\nu=\delta_{y}(\cdot)$
for two points $x,y\in\X$, we have $IS\W(\mu,\nu)=d(x,y)$, where
$d(\cdot,\cdot)$ is the spectral distance corresponding to the choice
of $\alpha(\cdot)$.
\end{prop}

\begin{proof}
We have $\phi_{\ell}\sharp\delta_{x}=\delta_{\phi_{\ell}(x)}$ and
similarly for $y$. Now $\W^{2}(\phi_{\ell}\sharp\mu,\phi_{\ell}\sharp\nu)=\W^{2}(\delta_{\phi_{\ell}(x)},\delta_{\phi_{\ell}(y)})=(\phi_{\ell}(x)-\phi_{\ell}(y))^{2}$.
This last equality follows from the fact that the 2-Wasserstein on
real line between delta measures is equal to the distance between
the two points. Then scaling and adding up gives exactly the kernel
distance $d(x,y)$ between the two points.
\end{proof}

\theoremstyle{plain}
\newtheorem*{prop44}{Proposition 4.4}
\begin{prop44}
\label{prop:MMD-equiv-1}Let $\D(\rho_{1},\rho_{2})=\vert\E_{x\sim\rho_{1}}[x]-\E_{y\sim\rho_{2}}[y]\vert$
for $\rho_{1},\rho_{2}\in\P(\R)$, then the corresponding intrinsic
sliced distance is equivalent to the MMD with the spectral kernel
$k(\cdot,\cdot)$.
\end{prop44}
\begin{proof}
We can rewrite the definition as follows:
\begin{align*}
IS\D^{2}(\mu,\nu)= & \sum_{\ell}\alpha(\lambda_{\ell})(\E_{x\sim\phi_{\ell}\sharp\mu}[x]-\E_{y\sim\phi_{\ell}\sharp\nu}[y])^{2}=\sum_{\ell}\alpha(\lambda_{\ell})(\E_{x\sim\mu}[\phi_{\ell}(x)]-\E_{y\sim\nu}[\phi_{\ell}(y)])^{2}\\
= & \sum_{\ell}\alpha(\lambda_{\ell})(\E_{x,x'\sim\mu}[\phi_{\ell}(x)\phi_{\ell}(x')]+\E_{y,y'\sim\nu}[\phi_{\ell}(y)\phi_{\ell}(y')]-2\E_{x\sim\mu,y\sim\nu}[\phi_{\ell}(x)\phi_{\ell}(y)])\\
= & \E_{x,x'\sim\mu}[\sum_{\ell}\alpha(\lambda_{\ell})\phi_{\ell}(x)\phi_{\ell}(x')]+\E_{y,y'\sim\nu}[\sum_{\ell}\alpha(\lambda_{\ell})\phi_{\ell}(y)\phi_{\ell}(y')]\\
 & -2\E_{x\sim\mu,y\sim\nu}[\sum_{\ell}\alpha(\lambda_{\ell})\phi_{\ell}(x)\phi_{\ell}(y)]\\
= & \E_{x,x'\sim\mu}[k(x,x')]+\E_{y,y'\sim\nu}[k(y,y')]-2\E_{x\sim\mu,y\sim\nu}[k(x,y)],
\end{align*}
where we used the spectral kernel $k(x,y)=\sum_{\ell}\alpha(\lambda_{\ell})\phi_{\ell}(x)\phi_{\ell}(y)$.
The last expression coincides with the MMD based on kernel $k(\cdot,\cdot)$;
see Lemma 6 in \cite{mmd}.
\end{proof}

\theoremstyle{plain}
\newtheorem*{prop45}{Proposition 4.5}
\begin{prop45}
\label{prop:stronger-than-MMD-1}
$MMD(\mu,\nu)\leq IS\W(\mu,\nu)$
when the same $\alpha(\cdot)$ is used in both constructions.
\end{prop45}
\begin{proof}
This follows directly from the fact that for $\rho_{1},\rho_{2}\in\P(\R)$
the inequality $\vert\E_{x\sim\rho_{1}}[x]-\E_{y\sim\rho_{2}}[y]\vert\leq\mathcal{W}_{1}(\rho_{1},\rho_{2})\leq\W(\rho_{1},\rho_{2})$
holds. Here the first inequality follows from the centroid bound \cite{GuibasEMD},
and the second inequality is the well-known ordering property of Wasserstein
distances \cite{VillaniBook}.
\end{proof}

\theoremstyle{plain}
\newtheorem*{thm46}{Theorem 4.6}
\begin{thm46}
\label{thm:metric-1}
If $\alpha(\lambda)>0$ for all $\lambda>0$, then $IS\W$ is a metric on $\P(\X)$.
\end{thm46}
\begin{proof}
In the light of the Proposition \ref{prop:hilbertian-1} it remains
only to prove that $IS\W(\mu,\nu)=0$ implies $\mu=\nu$. According
to Proposition \ref{prop:stronger-than-MMD-1}, $IS\W(\mu,\nu)=0$
yields $MMD(\mu,\nu)=0$. The assumption that $\alpha(\lambda)>0$
for all $\lambda>0$ implies that the spectral kernel $k(\cdot,\cdot)$
corresponding to $\alpha(\cdot)$ is universal \cite{10.5555/1248547.1248642}.
Universality implies the characteristic property \cite{mmd}, which
in turn means that $MMD(\mu,\nu)=0$ is equivalent to $\mu=\nu$,
proving the claim.
\end{proof}

\theoremstyle{plain}
\newtheorem*{prop47}{Proposition 4.7}
\begin{prop47}
\label{prop:lipschitz-1}
There exists a constant $c$ depending only
on $\X$ such that for all $\mu,\nu\in\P(\X)$ the inequality $IS\W(\mu,\nu)\leq c\W^{\X}(\mu,\nu)\sqrt{\sum_{\ell}\lambda_{\ell}^{(n+3)/2}\alpha(\lambda_{\ell})}$
holds; here, $n$ is the dimension of $\X$.
\end{prop47}
\begin{proof}
We remind $\W^{\X}$ is the 2-Wasserstein distance defined directly
$\P(\X)$ using the geodesic distance as the ground metric. The Neumann
eigenfunctions on compact manifolds satisfy the inequality $\Vert\nabla\phi_{\ell}\Vert_{\infty}\leq c_{1}\lambda_{\ell}\Vert\phi_{\ell}\Vert_{\infty}$,
see \cite{Hu2015}. Applying the bound used in the proof of convergence,
$\Vert\phi_{\ell}\Vert_{\infty}\leq c_{2}\lambda_{\ell}^{(n-1)/4}$,
we get that $\phi_{\ell}$ is Lipschitz with respect to the geodesic
distance on $\X$ with the Lipschitz constant bounded by $c\lambda_{\ell}\lambda_{\ell}^{(n-1)/4}=c\lambda_{\ell}^{(n+3)/4}$. 

Consider the optimal coupling between $\mu$ and $\nu$ whose cost
equals to $\W^{\X}(\mu,\nu)$. Note that this coupling straightforwardly
provides a coupling between the pushforwards $\phi_{\ell}\sharp\mu$
and $\phi_{\ell}\sharp\nu$. Using the Lipschitz property of eigenfunctions,
we see that the cost of the pushforward coupling is smaller than $c\lambda_{\ell}^{(n+3)/4}\W^{\X}(\mu,\nu)$.
Since any such coupling provides an upper bound on $\W(\phi_{\ell}\sharp\mu,\phi_{\ell}\sharp\nu)$,
we have $\W(\phi_{\ell}\sharp\mu,\phi_{\ell}\sharp\nu)\leq c\lambda_{\ell}^{(n+3)/4}\W^{\X}(\mu,\nu)$.
Plugging this into the formula for $IS\W$ we get the claimed bound.
\end{proof}
\theoremstyle{plain}
\newtheorem*{prop48}{Proposition 4.8}
\begin{prop48}
\label{prop:robust-1}
Let $\{\mu_{i}\}_{i=1}^{N}$ and $\{\nu_{i}\}_{i=1}^{N}$
be two collections of probability measures on $\P(\X)$, such that
$\forall i,\W^{\X}(\mu_{i},\nu_{i})\leq\epsilon$, then $\mathbb{T}(\{\mu_{i}\}_{i=1}^{N},\{\nu_{i}\}_{i=1}^{N})\leq C^{2}\epsilon^{2}$.
Here $C=c\sqrt{\sum_{\ell}\lambda_{\ell}^{(n+3)/2}\alpha(\lambda_{\ell})}$
from previous proposition and is assumed to be finite.
\end{prop48}
\begin{proof}
We have
\begin{align*}
\mathbb{T}(\{\mu_{i}\}_{i=1}^{N},\{\nu_{i}\}_{i=1}^{N})= & \left\Vert \frac{1}{N}\sum_{i=1}^{N}\eta(\mu_{i})-\frac{1}{N}\sum_{i=1}^{N}\eta(\nu_{i})\right\Vert _{\H}^{2}=\left\Vert \frac{1}{N}\sum_{i=1}^{N}(\eta(\mu_{i})-\eta(\nu_{i}))\right\Vert _{\H}^{2}\\
\leq & \frac{1}{N}\sum_{i=1}^{N}\Vert\eta(\mu_{i})-\eta(\nu_{i})\Vert_{\H}^{2}=\frac{1}{N}\sum_{i=1}^{N}IS\W^{2}(\mu_{i},\nu_{i})\leq\frac{1}{N}\sum_{i=1}^{N}(C\W^{\X}(\mu_{i},\nu_{i}))^{2}\\
\leq & \frac{1}{N}N(C\epsilon)^{2}=C^{2}\epsilon^{2}.
\end{align*}
\end{proof}

\subsection{Computational Details for Section~\ref{subsec:Approximate-Hilbert-Embedding}}
\label{subsec:app_compute}

The case of finite intervals is the building block for the general
case, so let us first consider the case of $\X=[0,T]$. We represent
a histogram over this interval by a discrete measure of the form $\mu=\sum w_{a}\delta_{x_{a}}$
with the histogram bin centers $x_{a}\in[0,T]$ and weights $w_{a}$
satisfying $\sum w_{a}=1$, where $a=1,2,...,A$. Note that it is
not required for the histograms in the collections to be supported
at the same bin locations. For a given histogram, let $\{x_{(a)},w_{(a)}\}_{a=1}^{A}$
be the locations sorted from smallest to largest and their corresponding
weights; since the bin locations are unique there will not be any
ties. The quantile function is computed via $F_{\mu}^{-1}(s):=\min\{x_{(a)}:\sum_{b\leq a}w_{(b)}>s\}$.
The approximate map $\eta_{D'}^{0}$ now can be computed using the
$s_{k}$-th quantile value $F_{\mu}^{\text{\textminus}1}(s_{k})$
for each value of $s_{k},k=1,...,D'$.

For a general domain $\X$, the histogram representation is the same
as above: $\sum w_{a}\delta_{x_{a}}$ with the histogram bin centers
$x_{a}\in\X$ and weights $w_{a}$ satisfying $\sum w_{a}=1$, where
$a=1,2,...,A$. The pushforward $\phi_{\ell}\sharp\mu$ gives a histogram
on the real line defined by $\sum w_{a}\delta_{\phi_{\ell}(x_{a})}$.
Note that while $x_{a}$ are distinct, their images under $\phi_{\ell}$
do not have to be distinct, so one re-aggregates the weights to obtain
$\sum_{a\in S}w'_{a}\delta_{\phi_{\ell}(x_{a})}$, where $S$ is a
subset of $1,2,...,A$ and $w'_{a}$ are the new weights. It is now
straightforward to compute the quantile function as before and build
the approximate map $(\eta_{D})_{\ell}$. Doing so for the different
values of $\ell$ and concatenating the resulting vectors gives $\eta_{D}$.

\begin{table}
\begin{centering}
\begin{tabular}{ccc}
\toprule 
$\X$ & Eigenvalues & Eigenfunctions\tabularnewline
\midrule 
$[0,T]$ & $(\frac{\pi\ell}{T})^{2}$ & $\sqrt{\frac{2}{T}}\cos\frac{\pi\ell x}{T}$\tabularnewline
$S^{1}(T)=[0,T]\mod T$ & $(\frac{2\pi\ell}{T})^{2}$ & $\sqrt{\frac{2}{T}}[\cos/\sin]\frac{2\pi\ell x}{T}$\tabularnewline
$[0,T_{1}]\times[0,T_{2}]$ & $(\frac{\pi\ell_{1}}{T_{1}})^{2}$+$(\frac{\pi\ell_{2}}{T_{2}})^{2}$ & $\sqrt{\frac{4}{T_{1}T_{2}}}\cos\frac{\pi\ell_{1}x}{T_{1}}\cos\frac{\pi\ell_{2}x}{T_{2}}$\tabularnewline
$S^{1}(T_{1})\times[0,T_{2}]$ & $(\frac{2\pi\ell_{1}}{T_{1}})^{2}+(\frac{\pi\ell_{2}}{T_{2}})^{2}$ & $\sqrt{\frac{4}{T_{1}T_{2}}}[\cos/\sin]\frac{2\pi\ell_{1}x}{T_{1}}\cos\frac{\pi\ell_{2}x}{T_{2}}$\tabularnewline
$S^{1}(T_{1})\times S^{1}(T_{2})$ & $(\frac{2\pi\ell_{1}}{T_{1}})^{2}+(\frac{2\pi\ell_{2}}{T_{2}})^{2}$ & $\sqrt{\frac{4}{T_{1}T_{2}}}[\cos/\sin]\frac{2\pi\ell_{1}x}{T_{1}}[\cos/\sin]\frac{\pi\ell_{2}x}{T_{2}}$\tabularnewline
$S^{2}$ & \multicolumn{2}{c}{Spherical harmonics \cite{spharmonics}}\tabularnewline
Graphs/Data Clouds/Meshes & \multicolumn{2}{c}{Eigen-decomposition of the Laplacian matrix}\tabularnewline
\bottomrule
\end{tabular}
\par\end{centering}
\caption{\label{tab:Eigenvalues-and-eigenfunctions}Eigenvalues and eigenfunctions
of the Laplace-Beltrami operator with Neumann boundary conditions
for simple manifolds. We exclude zero eigenvalue and the corresponding
constant eigenvector; thus, all indices $\ell,\ell_{1},\ell_{2}$
run over positive integers. The notation $[\cos/\sin]$ means picking
either the cosine or sine function\emph{---all choices must be used,
giving multiple eigenfunctions}.}
\end{table}

In practice, these computations can be carried out on a variety of
domains---analytic manifolds, manifolds discretized as point clouds
or meshes, and graphs. In most cases the spectral decomposition of
the Laplace-Beltrami operator or graph Laplacian has to be computed
numerically \cite{diffusion_map,eigencomp}. For applications that
involve simple manifolds, the eigenvalues and eigenfunctions can be
computed analytically. For completeness we list them in Table \ref{tab:Eigenvalues-and-eigenfunctions}.
Note that we benefit from the fact that the eigen-decomposition for
product spaces can be derived from the eigen-decompositions of the
components.

The choice of the function $\alpha(\cdot)$ determining the contributions
of each spectral band is problem specific. When working on manifolds
of low dimension, the choice of $\alpha(\cdot)$ that corresponds
to the biharmonic distance is convenient. While the diffusion distance
provides a general choice that works on manifolds of any dimension,
the biharmonic distance does not have any parameters to tune and was
shown to provide an excellent alternative to the geodesic distance
in low-dimensional settings \cite{Biharmonic}. When in doubt, inspecting
the behavior of the distance on the underlying domain will allow assessing
whether the distance is appropriate for the given problem.The importance
of relying on a well-behaved spectral distance was highlighted in
Proposition \ref{prop:ground-dist-1}.

\subsection{Proofs and Notes for Section~\ref{subsec:Resampling-Based-Test}}
\label{subsec:app_resamp_proofs}

We remind that we will be using the following test statistic for the
results that are discussed below:
\begin{equation}
\hat{\mathbb{T}}\equiv\sum_{i,j}\frac{IS\W^{2}(\mu_{i},\nu_{j})}{N_{1}N_{2}}-\sum_{i,j:i\neq j}\frac{IS\W^{2}(\mu_{i},\mu_{j})}{2N_{1}(N_{1}-1)}-\sum_{i,j:i\neq j}\frac{IS\W^{2}(\nu_{i},\nu_{j})}{2N_{2}(N_{2}-1)}.\label{eq:T-statistic-isd-1}
\end{equation}
\begin{prop}
\label{Prop-T-hat-limit-1} Assume conditions (i)-(iii) hold. Define
$N=N_{1}+N_{2}$, and assume that as $N_{1},N_{2}\rightarrow\infty$,
we have $N_{1}/N\rightarrow\rho_{1},N_{2}/N\rightarrow\rho_{2}=1-\rho_{1}$,
for some fixed $0<\rho_{1}<1$. Define a new measure $R$ as a scaled
mixture of the centered pushforward measures 
\[
R=\left(\frac{1}{\rho_{1}}+\frac{1}{\rho_{2}}\right)^{-1}\left[\frac{1}{\rho_{1}}\left(\eta\#P-\M_{\eta\#P}\right)+\frac{1}{\rho_{2}}\left(\eta\#Q-\M_{\eta\#Q}\right)\right]=\rho_{2}\left(\eta\#P-\M_{\eta\#P}\right)+\rho_{1}\left(\eta\#Q-\M_{\eta\#Q}\right).
\]
 Suppose $\gamma_{m},m=1,2,\ldots$ are the eigenvalues of 
\[
\frac{1}{\rho_{1}\rho_{2}}\int_{\H}\langle x,x'\rangle_{\H}\psi_{m}(x')dR(x')=\gamma_{m}\psi_{m}(x).
\]
Then under $H_{0}:\M_{\eta\#P}=\M_{\eta\#Q}$ we have
\begin{equation}
N\hat{\mathbb{T}}\leadsto\sum_{m=1}^{\infty}\gamma_{m}(A_{m}^{2}-1),\label{eq:H0lim-1}
\end{equation}
where $A_{m}$ are i.i.d. $\mathcal{N}(0,1)$ random variables. Under
$H_{1}:\M_{\eta\#P}\neq\M_{\eta\#Q}$ we have $\sqrt{N}(\hat{\mathbb{T}}-\mathbb{T})\leadsto N(0,\sigma_{1}^{2})$,
where 
\begin{eqnarray}
\sigma_{1}^{2} & = & 4\left[\frac{1}{\rho_{1}}\mathbb{V}_{\mu\sim P}\E_{\mu'\sim P}\langle\eta(\mu),\eta(\mu')\rangle_{\H}+\frac{1}{\rho_{2}}\mathbb{V}_{\nu\sim Q}\E_{\nu'\sim Q}\langle\eta(\nu),\eta(\nu')\rangle_{\H}+\right.\nonumber \\
 &  & \left.\frac{1}{\rho_{1}}\mathbb{V}_{\mu\sim P}\E_{\nu\sim Q}\langle\eta(\mu),\eta(\nu)\rangle_{\H}+\frac{1}{\rho_{2}}\mathbb{V}_{\nu\sim Q}\E_{\mu\sim P}\langle\eta(\mu),\eta(\nu)\rangle_{\H}\right].\label{eq:sigma1sq}
\end{eqnarray}
\end{prop}
\begin{proof}
Using the Hilbertianity of $IS\D$ (Proposition \ref{prop:hilbertian-1}),
we have 
\[
\begin{aligned}IS\D^{2}(\mu_{i},\mu_{j})= & \|\eta(\mu_{i})-\eta(\mu_{j})\|_{\H}^{2}\\
= & \|\eta(\mu_{i})\|_{\mathcal{\H}}^{2}+\|\eta(\mu_{j})\|_{\mathcal{\H}}^{2}-2\langle\eta(\mu_{i}),\eta(\mu_{j})\rangle_{\mathcal{\H}}
\end{aligned}
,
\]
Consequently
\[
\sum_{i,j:i\neq j}IS\D^{2}(\mu_{i},\mu_{j})=2(N_{1}-1)\sum_{i=1}^{N_{1}}\|\eta(\mu_{i})\|_{\H}^{2}-2\sum_{i,j:i\neq j}\langle\eta(\mu_{i}),\eta(\mu_{j})\rangle_{\H}.
\]
Similarly,
\begin{eqnarray*}
\sum_{i,j:i\neq j}IS\D^{2}(\nu_{i},\nu_{j}) & = & 2(N_{2}-1)\sum_{i=1}^{N_{2}}\|\eta(\nu_{i})\|_{\H}^{2}-2\sum_{i,j:i\neq j}\langle\eta(\nu_{i}),\eta(\nu_{j})\rangle_{\H},\\
\sum_{i,j}IS\D^{2}(\mu_{i},\nu_{j}) & = & N_{2}\sum_{i=1}^{N_{1}}\|\eta(\mu_{i})\|_{\H}^{2}+N_{1}\sum_{j=1}^{N_{2}}\|\eta(\nu_{j})\|_{\H}^{2}-2\sum_{i,j:i\neq j}\langle\eta(\mu_{i}),\eta(\nu_{j})\rangle_{\H}.
\end{eqnarray*}
Putting these back into Eq. (\ref{eq:T-statistic-isd-1}) after simplifying
and cancelling out the norm-square terms we have
\begin{align}
\hat{\mathbb{T}}= & \frac{1}{N_{1}(N_{1}-1)}\sum_{i,j:i\neq j}\langle\eta(\mu_{i}),\eta(\mu_{j})\rangle_{\H}+\frac{1}{N_{2}(N_{2}-1)}\sum_{i,j:i\neq j}\langle\eta(\nu_{i}),\eta(\nu_{j})\rangle_{\H}\nonumber \\
 & -\frac{2}{N_{1}N_{2}}\sum_{i,j}\langle\eta(\mu_{i}),\eta(\nu_{j})\rangle_{\H}.\label{eq:T-hat-simplify}
\end{align}
At this point, we replace the maps $\eta$ by their centered versions
$\tilde{\eta}(\mu)=\eta(\mu)-\M_{\eta\#P},\tilde{\eta}(\nu)=\eta(\nu)-\M_{\eta\#Q}$;
remember that the center of mass of $\eta\#P$ is denoted by $\M_{\eta\#P}$.
Accumulating the sample-level partial sums above the centering terms
cancel out under $H_{0}:\M_{\eta\#P}=\M_{\eta\#Q}$, so that each
$\eta$ can be replaced by $\tilde{\eta}$ in (\ref{eq:T-hat-simplify})
above. 

Denote $x_{i}\equiv\tilde{\eta}(\mu_{i}),y_{i}\equiv\tilde{\eta}(\nu_{i})$
as the Hilbert-embedded samples of $X\sim\tilde{\eta}\#P,Y\sim\tilde{\eta}\#Q$,
respectively. We remind now that $R$ is a mixture of the centered
pushforward measures: $R=\rho_{2}(\tilde{\eta}\#P)+\rho_{1}(\tilde{\eta}\#Q)$.
Let $L_{2}(\H,R)$ be the space of real-valued functions on $\H$
that are square integrable with respect to $R$. Now we can define
the following operator $S:L_{2}(\H,R)\rightarrow\H$,
\[
(Sf)(x):=\int_{\H}\langle x,x'\rangle_{\H}f(x')dR(x').
\]
Following condition (ii), $\langle\cdot,\cdot\rangle_{\H}$ is square-integrable
under $R$. The above operator is thus Hilbert-Schmidt, hence compact
\citep[Theorem VI.23]{ReedSimonBook}. Consequently, it permits an
eigenfunction decomposition with respect to measure $R$, $\langle x,x'\rangle_{\H}=\sum_{m=1}^{\infty}\gamma_{m}\psi_{m}(x)\psi_{m}(x')$,
for $x,x'\in\H$. Note that here $\psi_{m}:\H\rightarrow\R$ and 
\[
\int_{\H}\langle x,x'\rangle\psi_{m}(x')dR(x')=\gamma_{m}\psi_{m}(x),
\]

\[
\int_{\H}\psi_{m}(x)\psi_{n}(x)dR(x)=\delta_{mn}.
\]
Due to the centering of $\eta$ we also have when $\gamma_{m}\neq0$,
\[
\gamma_{m}\E_{X}[\psi_{m}(x)]=\int_{\H}\E_{X}[\langle x,x'\rangle_{\H}]\psi_{n}(x')dR(x')=0\quad\Rightarrow\quad\E_{X}[\psi_{m}(x)]=0.
\]
Similarly, $\E_{Y}[\psi_{m}(y)]=0.$ The V-statistic from the overall
sample can now be written as an infinite sum \citep[Section 5.5]{serflingbook}:
\[
\|\hat{\M}_{\eta\#P}-\hat{\M}_{\eta\#Q}\|_{\H}^{2}=\sum_{m=1}^{\infty}\gamma_{m}\left(\frac{1}{N_{1}}\sum_{i=1}^{N_{1}}\psi_{m}(x_{i})-\frac{1}{N_{2}}\sum_{i=1}^{N_{2}}\psi_{m}(y_{i})\right)^{2}:=\sum_{m=1}^{\infty}\gamma_{m}a_{m}^{2}.
\]
Our goal is to show that (a) $a_{m}\leadsto\mathcal{N}(0,(N\rho_{1}\rho_{2})^{-1})$,
for $\forall m$, and (b) $a_{m}$ and $a_{n}$ are independent when
$m\neq n.$

First note that
\[
\E(a_{m})=\E\left(\frac{1}{N_{1}}\sum_{i=1}^{N_{1}}\psi_{m}(x_{i})-\frac{1}{N_{2}}\sum_{i=1}^{N_{2}}\psi_{m}(y_{i})\right)=0.
\]
In addition we have,
\begin{eqnarray*}
Cov(a_{m},a_{n}) & = & \E(a_{m}a_{n})-\E(a_{m}).\E(a_{n})\\
 & = & \E(a_{m}a_{n})\\
 & = & \E\left(\frac{1}{N_{1}}\sum_{i=1}^{N_{1}}\psi_{m}(x_{i})-\frac{1}{N_{2}}\sum_{i=1}^{N_{2}}\psi_{m}(y_{i})\right)\left(\frac{1}{N_{1}}\sum_{i=1}^{N_{1}}\psi_{n}(x_{i})-\frac{1}{N_{2}}\sum_{i=1}^{N_{2}}\psi_{n}(y_{i})\right)\\
 & = & \E_{X}\left(\frac{1}{N_{1}^{2}}\sum_{i=1}^{N_{1}}\psi_{m}(x_{i})\psi_{n}(x_{i})\right)+\E_{Y}\left(\frac{1}{N_{2}^{2}}\sum_{i=1}^{N_{2}}\psi_{m}(y_{i})\psi_{n}(y_{i})\right)\\
 & = & \frac{1}{\rho_{1}N}\E_{X}\left(\frac{1}{N_{1}}\sum_{i=1}^{N_{1}}\psi_{m}(x_{i})\psi_{n}(x_{i})\right)+\frac{1}{\rho_{2}N}\E_{Y}\left(\frac{1}{N_{2}}\sum_{i=1}^{N_{2}}\psi_{m}(y_{i})\psi_{n}(y_{i})\right)\\
 & = & \frac{1}{N}\left[\frac{1}{\rho_{1}}\int_{\H}\psi_{m}(x)\psi_{n}(x)d(\tilde{\eta}\#P)(x)+\frac{1}{\rho_{2}}\int_{\H}\psi_{m}(y)\psi_{n}(y)d(\tilde{\eta}\#Q)(y)\right]\\
 & = & \frac{1}{N\rho_{1}\rho_{2}}\int_{\H}\psi_{m}(z)\psi_{n}(z)dR(z)\\
 & = & \frac{1}{N\rho_{1}\rho_{2}}\delta_{mn}.
\end{eqnarray*}
An application of CLT follows that (a) holds. This together with vanishing
covariance proves (b). Consequently, we can apply the CLT for degenerate
V-statistics \citep[Section 5.5.2]{serflingbook} to obtain the limiting
distribution, with $A_{m}\sim\mathcal{N}(0,1)$,
\[
N\|\hat{\M}_{\eta\#P}-\hat{\M}_{\eta\#Q}\|_{\H}^{2}\leadsto\sum_{m=1}^{\infty}\frac{\gamma_{m}}{\rho_{1}\rho_{2}}A_{m}^{2}.
\]

Let us now look at the difference between this V-statistic and our
U-statistic, i.e. $\hat{\mathbb{T}}$ in (\ref{eq:T-hat-simplify}).
We see that
\begin{eqnarray*}
\|\hat{\M}_{\eta\#P}-\hat{\M}_{\eta\#Q}\|_{\H}^{2}-\hat{\mathbb{T}} & = & \frac{1}{N_{1}^{2}}\sum_{i,j}\langle x_{i},x_{j}\rangle_{\H}+\frac{1}{N_{2}^{2}}\sum_{i,j}\langle y_{i},y_{j}\rangle_{\H}-\frac{2}{N_{1}N_{2}}\sum_{i,j}\langle x_{i},y_{j}\rangle_{\H}\\
 &  & -\frac{1}{N_{1}(N_{1}-1)}\sum_{i,j;i\neq j}\langle x_{i},x_{j}\rangle_{\H}+\frac{1}{N_{2}(N_{2}-1)}\sum_{i,j;i\neq j}\langle y_{i},y_{j}\rangle_{\H}+\frac{2}{N_{1}N_{2}}\sum_{i,j}\langle x_{i},y_{j}\rangle_{\H}\\
 & = & -\left[\frac{1}{N_{1}(N_{1}-1)}-\frac{1}{N_{1}^{2}}\right]\sum_{i,j;i\neq j}\langle x_{i},x_{j}\rangle_{\H}-\left[\frac{1}{N_{2}(N_{2}-1)}-\frac{1}{N_{2}^{2}}\right]\sum_{i,j;i\neq j}\langle y_{i},y_{j}\rangle_{\H}\\
 &  & +\left(\frac{1}{N_{1}^{2}}\sum_{i=1}^{N_{1}}\|x_{i}\|_{\H}^{2}+\frac{1}{N_{2}^{2}}\sum_{i=1}^{N_{2}}\|y_{i}\|_{\H}^{2}\right)\\
 & = & -K^{x}-K^{y}+B.
\end{eqnarray*}
We claim that $K^{x}=O_{p}(N_{1}^{-2}),K^{y}=O_{p}(N_{2}^{-2})$,
and $NB\stackrel{{P}}{\rightarrow}\sum_{m=1}^{\infty}\gamma_{m}(\rho_{1}\rho_{2})^{-1}$.
As a result,
\begin{eqnarray*}
N\left[\|\hat{\M}_{\eta\#P}-\hat{\M}_{\eta\#Q}\|_{\H}^{2}-\hat{\mathbb{T}}\right] & = & -NO_{p}(N_{1}^{-2})-NO_{p}(N_{2}^{-2})+\sum_{m=1}^{\infty}\frac{\gamma_{m}}{\rho_{1}\rho_{2}}+o_{p}(1)\\
 & = & \sum_{m=1}^{\infty}\frac{\gamma_{m}}{\rho_{1}\rho_{2}}+o_{p}(1),
\end{eqnarray*}
so that $N\hat{\mathbb{T}}\leadsto\sum_{m=1}^{\infty}\gamma_{m}(\rho_{1}\rho_{2})^{-1}(A_{m}^{2}-1)$,
and we conclude the proof by reassigning $\gamma_{m}\leftarrow\gamma_{m}(\rho_{1}\rho_{2})^{-1}$
to obtain (\ref{eq:H0lim-1}).

\textbf{Proof of Claim. }For the $K$-terms we have
\begin{eqnarray*}
K^{x} & = & \left[\frac{1}{N_{1}(N_{1}-1)}-\frac{1}{N_{1}^{2}}\right]\sum_{i,j;i\neq j}\langle x_{i},x_{j}\rangle_{\H}\\
 & = & \frac{1}{N_{1}^{2}(N_{1}-1)}\sum_{i,j;i\neq j}\langle x_{i},x_{j}\rangle_{\H}\\
 & = & \sum_{m=1}^{\infty}\gamma_{m}\frac{1}{N_{1}}\frac{1}{N_{1}(N_{1}-1)}\sum_{i,j;i\neq j}\psi_{m}(x_{i})\psi_{m}(x_{j})\\
 & = & \sum_{m=1}^{\infty}\gamma_{m}K_{m}^{x},
\end{eqnarray*}
where $K_{m}^{x}$ is defined as the inner sum. Since $\E_{X}\psi_{m}(x)=0$,
we have $\E_{X}(K_{m}^{x})=\frac{1}{N_{1}}[\E_{X}\psi_{m}(x)]^{2}=0$,
and
\begin{eqnarray}
Var_{X}(K_{m}^{x}) & = & \E_{X}[(K_{m}^{x})^{2}]\nonumber \\
 & = & \frac{1}{N_{1}^{2}}\E_{X}\left[\frac{1}{N_{1}^{2}(N_{1}-1)^{2}}\sum_{i\neq j}\sum_{l\neq k}\psi_{m}(x_{i})\psi_{m}(x_{j})\psi_{m}(x_{l})\psi_{m}(x_{k})\right]\label{eq:ln1}\\
 & = & \frac{1}{N_{1}^{2}}\E_{X}\left[\frac{1}{N_{1}^{2}(N_{1}-1)^{2}}\sum_{i\neq j}\psi_{m}^{2}(x_{i})\psi_{m}^{2}(x_{j})\right]\label{eq:ln2}\\
 & = & \frac{1}{N_{1}^{2}}.\frac{1}{N_{1}(N_{1}-1)}\left(\E_{X}[\psi_{m}^{2}(x)]\right)^{2}.\nonumber 
\end{eqnarray}
The cross terms---terms involving $l\neq i$ or $k\neq j$---vanish
due to the sample being iid and eigenfunctions having zero expectations.
The expectation in the last line is finite by assumption (ii), so
that $Var_{X}(K_{m}^{x})=O(N_{1}^{-4})$, giving $K_{m}^{x}=O_{p}(N_{1}^{-2})$.
Note that the assumption (ii) moreover implies the convergence of
the big-oh coefficients, leading to $K^{x}=\sum_{m=1}^{\infty}\gamma_{m}K_{m}^{x}=O_{p}(N_{1}^{-2})$.
Similarly we get $K^{y}=O_{p}(N_{2}^{-2})$.

For the term $B$, we have
\[
B=\frac{1}{N_{1}^{2}}\sum_{i=1}^{N_{1}}\|x_{i}\|_{\H}^{2}++\frac{1}{N_{2}^{2}}\sum_{i=1}^{N_{2}}\|y_{i}\|_{\H}^{2}=\sum_{m=1}^{\infty}\gamma_{m}\left[\frac{1}{N_{1}^{2}}\sum_{i=1}^{N_{1}}\psi_{m}^{2}(x_{i})+\frac{1}{N_{2}^{2}}\sum_{i=1}^{N_{2}}\psi_{m}^{2}(y_{i})\right]:=\sum_{m=1}^{\infty}\gamma_{m}C_{m}.
\]
Taking expectation,
\begin{eqnarray*}
\E_{X,Y}(C_{m}) & = & \frac{1}{\rho_{1}N}\int_{\H}\psi_{m}^{2}(x)d(\tilde{\eta}\#P)(x)+\frac{1}{\rho_{2}N}\int_{\H}\psi_{m}^{2}(y)d(\tilde{\eta}\#Q)(y)\\
 & = & \frac{1}{N\rho_{1}\rho_{2}}\int_{\H}\psi_{m}^{2}(z)dR(z)\\
 & = & \frac{1}{N\rho_{1}\rho_{2}}.
\end{eqnarray*}
Thus $\E_{X,Y}(NB)=\sum_{m}\gamma_{m}(\rho_{1}\rho_{2})^{-1}$. Finally,
\[
NB=\sum_{m=1}^{\infty}\gamma_{m}\left[\frac{1}{\rho_{1}N_{1}}\sum_{i=1}^{N_{1}}\psi_{m}^{2}(x_{i})+\frac{1}{\rho_{2}N_{2}}\sum_{i=1}^{N_{2}}\psi_{m}^{2}(y_{i})\right]\stackrel{{P}}{\rightarrow}\sum_{m=1}^{\infty}\gamma_{m}\left[\frac{1}{\rho_{1}}\E_{X}\psi_{m}^{2}(x)+\frac{1}{\rho_{2}}\E_{Y}\psi_{m}^{2}(y)\right]=\E_{X,Y}(NB)
\]
by the weak law of large numbers. This proves the claim for $B$.

\textbf{Alternative Distribution.} For the the limiting distribution
under $H_{1}$, notice that the first two terms in (\ref{eq:T-statistic-isd-1})
are the one-sample U-statistic calculated on the samples $\lbrace\mu_{i}\rbrace_{i=1}^{N_{1}}$
and $\lbrace\nu_{i}\rbrace_{i=1}^{N_{2}}$, respectively. Using the
CLT for non-degenerate U-statistics \citep[Section 5.5.1, Theorem A]{serflingbook},
we have
\begin{eqnarray*}
\sqrt{N_{1}}\left[\frac{\sum_{i,j:i\neq j}\langle\eta(\mu_{i}),\eta(\mu_{j})\rangle_{\H}}{N_{1}(N_{1}-1)}-\E_{\mu,\mu'\sim P}\langle\eta(\mu),\eta(\mu')\rangle_{\H}\right] & \leadsto & N\left(0,4\mathbb{V}_{\mu\sim P}\left[\E_{\mu'\sim P}\langle\eta(\mu),\eta(\mu')\rangle_{\H}\right]\right),\\
\sqrt{N_{2}}\left[\frac{\sum_{i,j:i\neq j}\langle\eta(\nu_{i}),\eta(\nu_{j})\rangle_{\H}}{N_{2}(N_{2}-1)}-\E_{\nu,\nu'\sim Q}\langle\eta(\nu),\eta(\nu')\rangle_{\H}\right] & \leadsto & N\left(0,4\mathbb{V}_{\nu\sim Q}\left[\E_{\nu'\sim Q}\langle\eta(\nu),\eta(\nu')\rangle_{\H}\right]\right).
\end{eqnarray*}
For the third summand, using an equivalent CLT for two-sample U-statistic
\citep[Theorem 2.1]{DEHLING2012124},
\begin{eqnarray*}
\sqrt{N}\left[\frac{\sum_{i,j}\langle\eta(\mu_{i}),\eta(\nu_{j})\rangle_{\H}}{N_{1}N_{2}}-\E_{\mu\sim P,\nu\sim Q}\langle\eta(\mu),\eta(\nu)\rangle_{\H}\right]\leadsto\\
N\left(0,\frac{1}{\rho_{1}}\mathbb{V}_{\mu\in P}\left[\E_{\nu\sim Q}\langle\eta(\mu),\eta(\nu)\rangle_{\H}\right]+\frac{1}{\rho_{2}}\mathbb{V}_{\nu\in Q}\left[\E_{\mu\sim P}\langle\eta(\mu),\eta(\nu)\rangle_{\H}\right]\right).
\end{eqnarray*}
We obtain (\ref{eq:sigma1sq}) by combining the above three results.
\end{proof}
The following result now ensures that approximations of $\hat{\mathbb{T}}$
using the top few eigenfunctions and a finite number of CDF embeddings
can be constructed with small approximation errors, provided the manifold
eigenvalues are declining suitably fast and the finite dimensional
$\eta_{D}(\cdot)$ is suitably smooth.
\begin{prop}
\label{Prop-LDbound-1}Suppose that (i), (ii) and (iii) hold. Then
we have $\sqrt{N}(\hat{\mathbb{T}}-\hat{\mathbb{T}}_{L_{N}})=o_{p}(1)$
and $\sqrt{N}(\hat{\mathbb{T}}_{L_{N}}-\tilde{\mathbb{T}}_{L_{N},D_{N}})=o_{p}(1)$
for the following choices of $L_{N},D_{N}$:

\[
L_{N}\geq\min_{L'}\left\lbrace L':\sum_{\ell=L'+1}^{\infty}\alpha_{\ell}\lambda_{\ell}^{(n+3)/2}\leq\frac{1}{N^{1+\delta}}\right\rbrace ,\quad D_{N}\geq kc^{2}N^{1+\delta}\sum_{l=1}^{L_{N}}\alpha_{\ell}\lambda_{\ell}^{(n-1)/2},
\]
where $\delta,k>0$ are constants depending only on $\X$.
\end{prop}
As we mention in the discussion after condition (i), for the heat
kernel with tuning parameter $t$: $\alpha(\lambda)=\exp(-t\lambda)$,
the assumption (i) that $\sum_{\ell=1}^{\infty}\alpha_{l}\lambda_{\ell}^{(n+3)/2}<\infty$
holds. The bound on $D_{N}$ is a consequence of classical bounds
on Riemann sum approximation errors in terms of $\|\eta'\|{}_{\infty}$.
Absolute continuity of $\mu\sim P,\nu\sim Q$ ensures the existence
of $(F_{\phi_{\ell}\sharp\mu}^{-1})'(s),(F_{\phi_{\ell}\sharp\nu}^{-1})'(s)$
(where prime denotes the derivative) for Lebesgue-almost every $s\in[0,1]$
\citep[Lemma 2.3]{GavishEtal19}.
\begin{proof}
Notice that given $L_{N}$, summands in the expression $\hat{\mathbb{T}}-\hat{\mathbb{T}}_{L_{N}}$
are the tail sums $\sum_{\ell=L_{N}+1}^{\infty}\alpha_{l}\W^{2}(\phi_{\ell}\sharp\cdot,\phi_{\ell}\sharp\cdot)$
starting at the $L_{N}+1^{\text{th}}$ term. Using a similar approach
as the proof of Proposition \ref{prop:lipschitz-1}, this is bounded
above by a scalar multiple of the geodesic distance, specifically
$c\W^{\X}(\cdot,\cdot)\sqrt{\sum_{\ell=L_{N}+1}^{\infty}\alpha_{\ell}\lambda_{\ell}^{(n+3)/2}}$.
By assumption $\sum_{\ell=1}^{\infty}\alpha_{\ell}\lambda_{\ell}^{(n+3)/2}<\infty$,
so that given $\epsilon>0$ we can always choose a starting point
to make the tail sum $<\epsilon$. The choice of $L_{N}$ follows
by taking $\epsilon=N^{-(1+\delta)}$.

To obtain the choice of $D_{N}$, we first use a similar approach
to the proof of Proposition \ref{Prop-T-hat-limit-1} to simplify
$\tilde{\mathbb{T}}_{L,D'}$ for any $L,D'$:
\begin{align}
\tilde{\mathbb{T}}_{L,D'}=\sum_{\ell=1}^{L} & \left[\frac{1}{N_{1}(N_{1}-1)}\sum_{i,j:i\neq j}\eta_{D'}(\phi_{\ell}\sharp\mu_{i})^{T}\eta_{D'}(\phi_{\ell}\sharp\mu_{j})+\frac{1}{N_{2}(N_{2}-1)}\sum_{i,j:i\neq j}\eta_{D'}(\phi_{\ell}\sharp\nu_{i})^{T}\eta_{D'}(\phi_{\ell}\sharp\nu_{j})\right.\nonumber \\
 & \left.-\frac{2}{N_{1}N_{2}}\sum_{i,j}\eta_{D'}(\phi_{\ell}\sharp\mu_{i})^{T}\eta_{D'}(\phi_{\ell}\sharp\nu_{j})\right].\label{eq:T-tilde-simplify}
\end{align}
Recall that the inverse CDF transformation induced by $\eta_{0}(\phi_{\ell}\sharp\mu)\equiv F_{\phi_{\ell}\sharp\mu}^{-1}$
maps $[0,1]$ to a bounded interval that is the range of $\phi_{\ell}$,
and $\Vert\phi_{\ell}\Vert_{\infty}\leq c\lambda_{\ell}^{(n-1)/4}$
using H\"ormander's bound on the supremum norm of the eigenfunctions.
Using classical results on Riemann sum approximation errors \cite{riemann1,riemann2}
we thus have for any $\ell$:
\[
\left|\alpha_{\ell}\langle\eta_{0}(\phi_{\ell}\sharp\mu),\eta_{0}(\phi_{\ell}\sharp\nu)\rangle_{\H}-\eta_{D'}(\phi_{\ell}\sharp\mu)^{T}\eta_{D'}(\phi_{\ell}\sharp\nu)\right|\leq\frac{k}{D'}\alpha_{\ell}\left\Vert (F_{\phi_{\ell}\sharp\mu}^{-1}F_{\phi_{\ell}\sharp\nu}^{-1})'\right\Vert _{\infty}\leq\frac{2kc^{2}}{D'}\alpha_{\ell}\lambda_{\ell}^{(n-1)/2}.
\]
Given $L=L_{N}$, we simply choose $D'=D_{N}$ large enough to make
the right hand side above smaller than $N^{-(1+\delta)}$. While it
is possible to make the upper bound tighter using recent results (such
as \cite{riemann2}), the above coarser bound suffices for our purpose.
\end{proof}
We now state a version of Theorem~\ref{Thm:asy-tilde}, with specifications for $\gamma_{m},\sigma_{1}^{2},L_{N},D_{N}$
now available through the above two results.
\begin{thm}
\label{Thm:asy-tilde-1} Assume conditions (i)-(iii) hold. Define
$N=N_{1}+N_{2}$, and suppose that as $N_{1},N_{2}\rightarrow\infty$,
we have $N_{1}/N\rightarrow\rho_{1},N_{2}/N\rightarrow\rho_{2}=1-\rho_{1}$,
for some fixed $0<\rho_{1}<1$.With $L\geq L_{N},D'\geq D_{N}$ chosen
per Proposition \ref{Prop-LDbound-1}, under $H_{0}:\M_{\eta\#P}=\M_{\eta\#Q}$
we have
\[
N\tilde{\mathbb{T}}_{L,D'}\leadsto\sum_{m=1}^{\infty}\gamma_{m}(A_{m}^{2}-1),
\]
where $A_{m},\gamma_{m}$ are defined as in Proposition \ref{Prop-T-hat-limit-1}.
Further, under $H_{1}:\M_{\eta\#P}\neq\M_{\eta\#Q}$ we have $\sqrt{N}\left(\tilde{\mathbb{T}}_{L,D'}-\mathbb{T}\right)\leadsto N(0,\sigma_{1}^{2})$.
\end{thm}
\begin{proof}
This a combination of Propositions \ref{Prop-T-hat-limit-1} and \ref{Prop-LDbound-1},
and Slutsky's theorem. 
\end{proof}

We conclude with a proof of a specified version of Theorem~\ref{thm:perm-power}, which gives power guarantee of the test based on $\tilde{\mathbb{\mathbb{T}}}_{L,D'}$
for contiguous alternatives.
\begin{thm}
\label{thm:perm-power-1}
Assume conditions (i)-(iii) hold, and let
$L,D'$ be chosen as in Theorem \ref{Thm:asy-tilde-1}. Then for the
sequence of contiguous alternatives $H_{1N}$ such that $N\|\delta_{N}\|_{\mathcal{H}}^{2}\rightarrow\infty$,
the test based on $\tilde{\mathbb{T}}_{L,D'}$ is consistent for any
$\alpha\in(0,1)$, that is as $N\rightarrow\infty$ the asymptotic
power approaches 1.
\end{thm}
\begin{proof}
It is enough the prove consistency using $\hat{\mathbb{T}}$, as the
difference between $\hat{\mathbb{T}}$ and $\tilde{\mathbb{T}}_{L,D'}$
is negligible by choice of $L,D'$. To do so we utilize proof techniques
similar to Theorem 13 in \citet{mmd}. Define $c_{N}:=N^{1/2}\|\delta_{N}\|_{\mathcal{H}}$,
and expand the simplified centered version of the test statistic in
(\ref{eq:T-hat-simplify}) but under $H_{1}$ so that the centering
terms do not cancel out:
\begin{align}
\hat{\mathbb{T}}_{c}= & \frac{1}{N_{1}(N_{1}-1)}\sum_{i,j:i\neq j}\langle\eta(\mu_{i})-\M_{\eta\#P},\eta(\mu_{j})-\M_{\eta\#P}\rangle_{\H}\nonumber \\
 & +\frac{1}{N_{2}(N_{2}-1)}\sum_{i,j:i\neq j}\langle\eta(\nu_{i})-\M_{\eta\#Q},\eta(\nu_{j})-\M_{\eta\#Q}\rangle_{\H}\label{eq:T-statistic-centered}\\
 & \left.-\frac{2}{N_{1}N_{2}}\sum_{i,j}\langle\eta(\mu_{i})-\M_{\eta\#P},\eta(\nu_{j})-\M_{\eta\#Q}\rangle_{\H}\right].\nonumber 
\end{align}
The centered pushforwards have the same Hilbert centroids, thus as
$N\rightarrow\infty$ by Proposition \ref{Prop-T-hat-limit-1},
\[
N\hat{\mathbb{T}}_{c}\leadsto\sum_{m=1}^{\infty}\gamma_{m}(A_{m}^{2}-1):=S.
\]
Subtracting $\hat{\mathbb{T}}_{c}$ from $\hat{\mathbb{T}}$ and its
expansion in Eq. (\ref{eq:T-statistic-isd-1}) on the left and right
hand respectively, then simplifying we have
\begin{eqnarray}
N(\hat{\mathbb{T}}-\hat{\mathbb{T}}_{c}) & =N & \left[-\frac{1}{N_{1}}\sum_{i=1}^{N_{1}}\langle\delta_{N},\eta(\mu_{i})-\M_{\eta\#P}\rangle_{\H}+\frac{1}{N_{2}}\sum_{i=1}^{N_{2}}\langle\delta_{N},\eta(\nu_{i})-\M_{\eta\#Q}\rangle_{\H}+\frac{\langle\delta_{N},\delta_{N}\rangle_{\mathcal{H}}}{2}\right]\nonumber \\
 & =N & \left[\frac{\|\delta_{N}\|_{\mathcal{H}}}{N_{1}}\sum_{i=1}^{N_{1}}\left\langle \frac{\delta_{N}}{\|\delta_{N}\|_{\mathcal{H}}},\eta(\mu_{i})-\M_{\eta\#P}\right\rangle _{\mathcal{H}}\right.\nonumber \\
 &  & \left.-\frac{\|\delta_{N}\|_{\mathcal{H}}}{N_{2}}\sum_{i=1}^{N_{2}}\left\langle \frac{\delta_{N}}{\|\delta_{N}\|_{\mathcal{H}}},\eta(\nu_{i})-\M_{\eta\#Q}\right\rangle _{\mathcal{H}}+\frac{\|\delta_{N}\|_{\mathcal{H}}^{2}}{2}\right].\label{eq:diff-hat}
\end{eqnarray}
Given $N$ the inner products $\langle\delta_{N}/\|\delta_{N}\|_{\mathcal{H}},\eta(\mu_{i})-\M_{\eta\#P}\rangle_{\mathcal{H}}$
are i.i.d. random variables with mean 0, so by CLT then using $\|\delta_{N}\|_{\mathcal{H}}=c_{N}N^{-1/2}$
we get 
\[
\frac{1}{\sqrt{N_{1}}}\sum_{i=1}^{N_{1}}\left\langle \frac{\delta_{N}}{\|\delta_{N}\|_{\mathcal{H}}},\eta(\mu_{i})-\M_{\eta\#P}\right\rangle _{\mathcal{H}}\leadsto U\quad\Rightarrow\quad\frac{N\|\delta_{N}\|_{\mathcal{H}}}{N_{1}}\sum_{i=1}^{N_{2}}\left\langle \frac{\delta_{N}}{\|\delta_{N}\|_{\mathcal{H}}},\eta(\nu_{i})-\M_{\eta\#Q}\right\rangle _{\mathcal{H}}\leadsto\frac{c_{N}}{\sqrt{\rho_{1}}}U,
\]
where $U$ is the zero mean Gaussian random variable that is the limiting
distribution of the above inner product sum. Similarly we have
\[
\frac{N\|\delta_{N}\|_{\mathcal{H}}}{N_{2}}\sum_{i=1}^{N_{2}}\left\langle \frac{\delta_{N}}{\|\delta_{N}\|_{\mathcal{H}}},\eta(\nu_{i})-\M_{\eta\#Q}\right\rangle _{\mathcal{H}}\leadsto\frac{c_{N}}{\sqrt{\rho_{2}}}V,
\]
where $V$ is also Gaussian, zero mean, and independent of $U$. Putting
everything together in the right hand side of (\ref{eq:diff-hat}),
and using $\|\delta_{N}\|_{\mathcal{H}}=c_{N}N^{-1/2}$, given the
threshold $t_{\alpha}$ for a level-$\alpha$ test 
\[
P_{H_{N}}\left(N\hat{\mathbb{T}}>t_{\alpha}\right)\rightarrow P\left[S+c_{N}\left(\frac{U}{\sqrt{\rho_{1}}}-\frac{V}{\sqrt{\rho_{2}}}\right)+\frac{c_{N}^{2}}{2}>t_{\alpha}\right].
\]
By assumption $c_{N}^{2}\rightarrow\infty$, so the asymptotic power
approaches 1 as $N\rightarrow\infty$.
\end{proof}

\subsection{Proofs and Notes for Section \ref{subsec:Combination-Approach}}
\label{subsec:app_comb_proofs}

To guarantee size control when using the the harmonic mean $p$-value
we establish a version of Theorem 1 from \citet{CauchyP}. Assume that
a test statistic $Z\in\R^{D}$ has null distribution with zero mean
and every pair of coordinates of $Z$ follows bivariate Gaussian distribution.
Compute the coordinate-wise two-sided $p$-values $p_{k}=2(1-\Phi(\vert Z_{k}\vert))$
where $\Phi$ is the standard Gaussian CDF. 
\theoremstyle{plain}
\newtheorem*{thm54}{Theorem 5.4}
\begin{thm54}
Let $p_{k},k=1,...,D$ be the null $p$-values as above and $p^{H}$
computed via harmonic mean approach, then 
\[
\lim_{\alpha\to0}\frac{\mathrm{Prob}\{p^{H}\leq\alpha\}}{\alpha}=1.
\]
\end{thm54}
\begin{proof}
The proof of Theorem 1 from \cite{CauchyP} hinges on Lemma 3 in their
supplemental material. We show that Lemma 3 holds for the harmonic
mean combination method. Note that the multiplication by $\pi$ present
in Lemma 3 cancels out when inverse cotangent with a multiplier of
$1/\pi$ is applied later on; so it is not relevant to the flow of
the proof. 

To this end, consider the functions $p(x)=2(1-\Phi(\vert x\vert))$
and $h(x)=1/p(x)$. We need to prove the following three statements: 

(1) for any $|x|>\Phi^{-1}(3/4)$,
\[
\frac{\cos[p(x)\pi]}{p(x)}\leq h(x)\leq\frac{1}{p(x)}
\]

(2) For any constant $0<|a|<1$, we have 
\[
\lim_{x\to+\infty}\frac{h(x)}{x^{2}h(ax)}>c_{a}>0,
\]
where $c_{a}$ is some constant only dependent on $a$.

(3) Suppose that $X_{0}$ has standard normal distribution, then we
have 
\[
P\{h(X_{0})\geq t\}=\frac{1}{t}+O(1/t^{3}).
\]

Statement (1) is trivial, as $h(x)=1/p(x)$ by definition and the
cosine function is upper bounded by one. Statement (2) holds by the
same argument as in the supplement of \cite{CauchyP}. Statement (3)
follows from the fact that when $X_{0}$ is standard normal, then
$p(x)$ is a null $p$-value, and so 
\[
P\{h(X_{0})\geq t\}=P\{p(X_{0})\leq1/t\}=\frac{1}{t}.
\]
Note that there is no $O(1/t^{3})$ term at all, but we kept the form
of the statement the same as in \cite{CauchyP}.

Now, the proof of Theorem 1 from \cite{CauchyP} with weights $\omega_{k}=1/D,k=1,2,...,D$
goes through to give 
\[
P\left\{ \frac{1}{D}\sum\frac{1}{p_{k}}\geq t\right\} =\frac{1}{t}+o(1/t).
\]
Note that $p^{H}=H\left(D/(\frac{1}{p_{1}}+\frac{1}{p_{2}}+\cdots+\frac{1}{p_{D}})\right)$,
where the function $H$ has a known form described in \cite{HarmonicP}
and satisfies $H(x)/x\to1$ as $x\to0$. Thus, as $\alpha\to0$, we
have 
\[
P\left\{ p^{H}\leq\alpha\right\} \asymp P\left\{ \frac{1}{D}\sum\frac{1}{p_{k}}\geq1/\alpha\right\} \asymp\frac{1}{1/\alpha}+o(\frac{1}{1/\alpha})\asymp\alpha.
\]
\end{proof}

\section{Details of numerical experiments}
\label{sec:appB}

\subsection{Synthetic data}
\label{subsec:app_simulations}
We compare the performance of our tests on data from a number of domains
with several existing methods, and settings of the embedding parameters
$L,D'$. For evaluation, we use empirical power at different degrees
of departure from the null hypothesis, calculated by averaging the
proportion of rejections at level $\alpha=0.05$ over 1000 independent
datasets with samples divided into two groups of sizes $n_{1}=60,n_{2}=40$.
To ensure the tests are well-calibrated, we also calculate nominal
sizes assuming the two sample groups are drawn from the same random
meta-distribution. We calculate eigenvalues and eigenfunctions using
analytical expressions provided in Table~\ref{tab:Eigenvalues-and-eigenfunctions}. We fix $\alpha(\lambda)=e^{-\lambda}$
(i.e. heat kernel with $t=1$) for all experiments, in order to avoid unfair advantage from tuning this parameter when comparing to baselines. Also, when using the p-value combination test each t-test scales the mean difference by standard deviation, so the weight functions cancel out anyway. In general, when $t$ is small, the $\hat{\mathbb{T}}$ statistic is more democratic between low and high frequency eigenfunctions, allowing to capture finer details of the underlying domain (assuming large enough 
$L$). When $t$ is large, $\hat{\mathbb{T}}$ focuses on low frequency eigenfunctions so more global differences dominate the computation of 
$\hat{\mathbb{T}}$.

\paragraph*{Finite intervals}

To obtain our base measures $\mu_{i},\nu_{i},$ we generate bin probabilities
as (shifted and normalized) values of the function $f(t_{j})=\mu(t_{j})+\alpha(t_{j})$
at $m=30$ fixed design points $t_{j}=j/(m+1),j=\{1,2,\ldots,m\}$,
and
\begin{eqnarray*}
\mu(t_{j}) & = & 1.2+2.3\cos(2\pi t_{j})+4.2\sin(2\pi t_{j}),\\
\alpha(t_{j}) & = & \epsilon_{0}+\sqrt{2}\epsilon_{1}\cos(2\pi t_{j})+\sqrt{3}\epsilon_{2}\sin(2\pi t_{j}),
\end{eqnarray*}
where $\epsilon_{0,}\epsilon_{1},\epsilon_{2}\sim N(0,1)$ clipped
between $[-3,3]$. Group 1 and 2 samples are obtained as $\mu_{i}(\cdot)\equiv f(\cdot)$
and $\nu_{i}(\cdot)\equiv f(\cdot)+\delta$ respectively, where $\delta\in[0,4]$
is a constant. To make the sample functions non-negative, we shift
all functions by $M=3(1+\sqrt{2}+\sqrt{3})$. Finally, as the $m$-length
vector of bin counts for a sample, we generate a random vector from
the Multinomial distribution with 1000 trials, $m$ outcomes and the
outcome probabilities proportional to the shifted functional observations
corresponding to that sample. 

We use embedding dimensions $L=3,D'=10$ to compare our method against
11 functional ANOVA tests---for brevity we report results for 3 of
them which use different methodological approaches (see Appendix for
complete results). All methods maintain nominal size for $\delta=0$
(Figure \ref{fig:power} a). While the combination test (ISD comb)
based on our proposal outperformed all the other tests across all
values of $\delta$, the bootstrap test that uses the overall $\mathbb{T}$
statistic (ISD T boot) performs better than Fmaxb but worse than others.
Table \ref{tab:func-outputs} shows the outputs for the other 8 competing
methods from the R package \texttt{fdANOVA} for the finite intervals
synthetic data setting\footnote{See \url{https://www.rdocumentation.org/packages/fdANOVA/versions/0.1.2/topics/fanova.tests} for full names of all methods.}.

\setlength{\tabcolsep}{2pt}

\begin{figure}[t]
\begin{centering}
\begin{tabular}{cccc}
\includegraphics[width=0.24\textwidth]{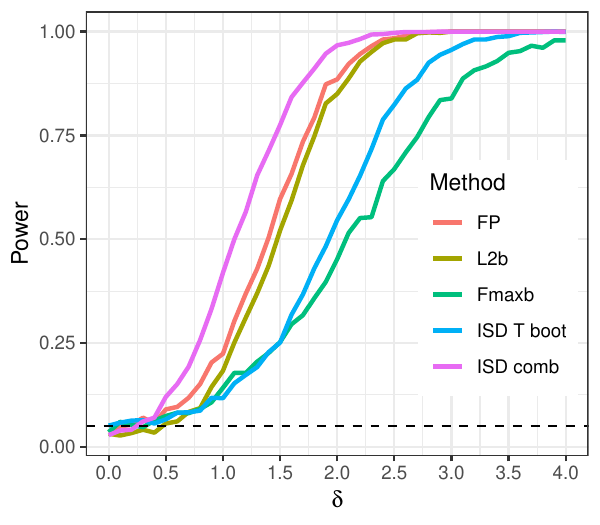} & \includegraphics[width=0.24\textwidth]{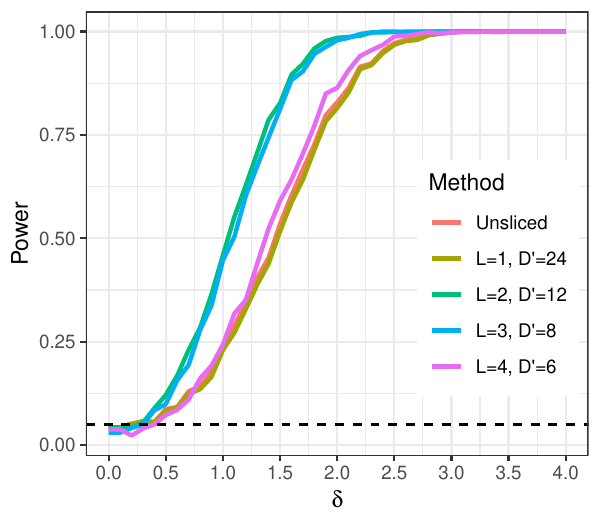} & \includegraphics[width=0.24\textwidth]{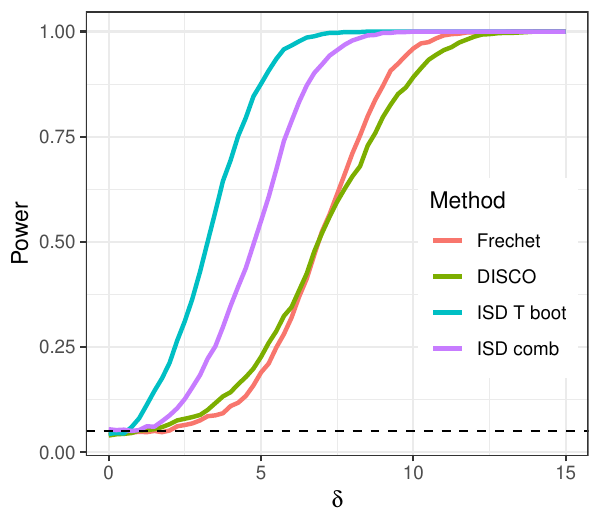} & \includegraphics[width=0.24\textwidth]{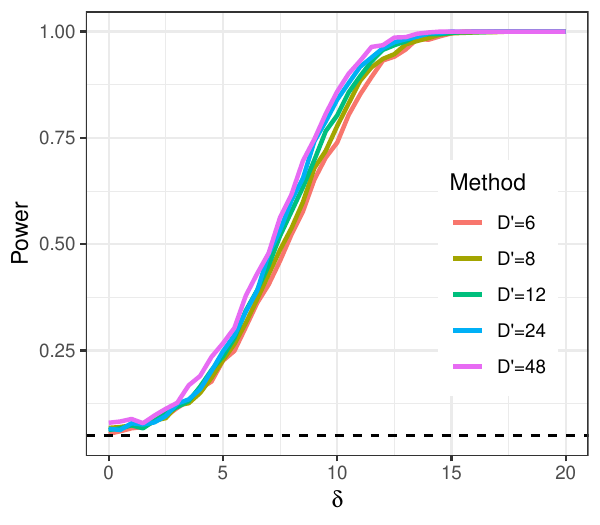}\tabularnewline
(a) & (b) & (c) & (d)\tabularnewline
\end{tabular}
\par\end{centering}
\caption{
\label{fig:power-1}
Performance on synthetic finite interval and manifold
data. Finite interval: (a) comparison with existing methods---a test
based on basis function representation (FP) \cite{FPtest}, a sum-type
$\ell_{2}$ norm-based test (L2b) \cite{L2btest}, and a max-type
test \cite{Fmaxbtest} that uses the maximum of coordinate-wise $F$
statistic (Fmaxb); (b) unsliced vs. different settings of $(L,D')$.
Manifold data: (c) circular data, comparing with Fr\'echet ANOVA \cite{10.1093/biomet/asz052},
and the DISCO nonparametric test \cite{disco}; (d) harmonic combination
tests on cylindrical data for $L=4$. Dotted lines indicates nominal
size of all tests ($\alpha=0.05$).}
\end{figure}

\setlength{\tabcolsep}{10pt}

\begin{table}[ht!]
\centering{}%
\scalebox{1}{
\begin{tabular}{ccccccccc}
\toprule 
$\delta$ & CH & CS & L2N & L2b & FN & FB & Fb & GPF\tabularnewline
\midrule 
0 & 0.031 & 0.03 & 0.033 & 0.024 & 0.031 & 0.028 & 0.033 & 0.026\tabularnewline
0.1 & 0.025 & 0.024 & 0.03 & 0.044 & 0.027 & 0.03 & 0.041 & 0.021\tabularnewline
0.2 & 0.026 & 0.029 & 0.037 & 0.06 & 0.033 & 0.034 & 0.058 & 0.025\tabularnewline
0.3 & 0.036 & 0.041 & 0.044 & 0.067 & 0.041 & 0.04 & 0.067 & 0.033\tabularnewline
0.4 & 0.034 & 0.035 & 0.036 & 0.057 & 0.034 & 0.035 & 0.056 & 0.032\tabularnewline
0.5 & 0.051 & 0.052 & 0.058 & 0.091 & 0.056 & 0.057 & 0.088 & 0.044\tabularnewline
0.6 & 0.056 & 0.066 & 0.066 & 0.089 & 0.061 & 0.066 & 0.088 & 0.051\tabularnewline
0.7 & 0.07 & 0.083 & 0.083 & 0.121 & 0.084 & 0.081 & 0.119 & 0.064\tabularnewline
0.8 & 0.085 & 0.097 & 0.095 & 0.151 & 0.093 & 0.094 & 0.144 & 0.081\tabularnewline
0.9 & 0.118 & 0.142 & 0.14 & 0.2 & 0.144 & 0.137 & 0.194 & 0.118\tabularnewline
1 & 0.158 & 0.182 & 0.176 & 0.232 & 0.183 & 0.173 & 0.228 & 0.154\tabularnewline
1.1 & 0.215 & 0.247 & 0.246 & 0.303 & 0.251 & 0.242 & 0.301 & 0.212\tabularnewline
1.2 & 0.27 & 0.31 & 0.303 & 0.375 & 0.311 & 0.3 & 0.368 & 0.27\tabularnewline
1.3 & 0.328 & 0.363 & 0.357 & 0.438 & 0.37 & 0.353 & 0.43 & 0.324\tabularnewline
1.4 & 0.395 & 0.432 & 0.432 & 0.504 & 0.436 & 0.423 & 0.499 & 0.394\tabularnewline
1.5 & 0.488 & 0.52 & 0.514 & 0.592 & 0.521 & 0.511 & 0.586 & 0.483\tabularnewline
1.6 & 0.534 & 0.595 & 0.576 & 0.652 & 0.593 & 0.566 & 0.647 & 0.544\tabularnewline
1.7 & 0.628 & 0.677 & 0.669 & 0.723 & 0.678 & 0.661 & 0.719 & 0.631\tabularnewline
1.8 & 0.704 & 0.737 & 0.727 & 0.789 & 0.748 & 0.725 & 0.785 & 0.707\tabularnewline
1.9 & 0.785 & 0.823 & 0.812 & 0.869 & 0.827 & 0.806 & 0.867 & 0.793\tabularnewline
2 & 0.83 & 0.849 & 0.844 & 0.88 & 0.85 & 0.841 & 0.875 & 0.832\tabularnewline
2.1 & 0.865 & 0.888 & 0.881 & 0.916 & 0.887 & 0.878 & 0.915 & 0.872\tabularnewline
2.2 & 0.903 & 0.922 & 0.916 & 0.946 & 0.928 & 0.912 & 0.946 & 0.907\tabularnewline
2.3 & 0.938 & 0.95 & 0.944 & 0.964 & 0.951 & 0.944 & 0.963 & 0.944\tabularnewline
2.4 & 0.958 & 0.973 & 0.967 & 0.977 & 0.972 & 0.966 & 0.976 & 0.964\tabularnewline
2.5 & 0.974 & 0.98 & 0.976 & 0.985 & 0.981 & 0.975 & 0.985 & 0.974\tabularnewline
2.6 & 0.977 & 0.981 & 0.979 & 0.987 & 0.981 & 0.978 & 0.986 & 0.977\tabularnewline
2.7 & 0.989 & 0.996 & 0.992 & 0.997 & 0.996 & 0.992 & 0.997 & 0.991\tabularnewline
2.8 & 0.997 & 0.998 & 0.997 & 0.998 & 0.998 & 0.997 & 0.998 & 0.996\tabularnewline
2.9 & 0.996 & 0.997 & 0.996 & 0.999 & 0.997 & 0.996 & 0.999 & 0.997\tabularnewline
3 & 0.998 & 1 & 0.999 & 1 & 1 & 0.999 & 1 & 0.999\tabularnewline
\bottomrule
\end{tabular}
}
\caption{\label{tab:func-outputs}Outputs for other methods in the functional
curves synthetic data setting.}
\end{table}

We also compare the $p$-value combination test based on an \emph{unsliced}
24-dimensional inverse CDF embedding with sliced $IS\W$-based tests
(Figure \ref{fig:power} b). We use multiple pairs of $(L,D')$ values,
all of them giving overall embeddings of dimension $D=LD'=24$. The
performance of an $IS\W$-based test that uses slicing over only the
first eigenfunction is almost as good as the unsliced version. With
more eigenfunctions, the powers first improve considerably, then become
similar to the unsliced version again.

\paragraph*{Manifold domains}

We consider data from distributions on circles and cylinders. For
circular data, we take von Mises distributions with randomly chosen
parameters as our samples. For an angle $x$ (measured in radians),
the von Mises probability density function is given by $f(x|\mu,\kappa)=\exp[\kappa\cos(x-\mu)](2\pi I_{0}(\kappa))^{-1}$,
where $I_{0}(\kappa)$ is the modified Bessel function of order 0.
We fix $\kappa=2$, and use $\mu\equiv\mu_{i}\sim N(0,0.1^{2})$,
$\mu\equiv\nu_{i}\sim N(\delta,0.1^{2})$ for samples from group 1
and 2 respectively---with $\delta\in[0,15]\times\pi/180$ (i.e. 0
to 15 degrees converted to radians). As each observation vector, we
take 100 random draws from each sample-specific distribution. For
our embeddings, we use $L=10,D'=20$, and so our final embedding dimension
is $10\times20\times2=400$. Since the competing methods cannot handle
circular geometry directly, to implement them we cut the circle into
an interval. Figure \ref{fig:power} (c) shows that all methods maintain
nominal size, but both our tests maintain considerably higher power
than existing methods for all $\delta$.

We generate cylindrical data in the form of samples of a bivariate
random vector $(\Theta,X)$, using the cylindrical density function
proposed by \cite{MardiaSutton}:

\[
f(\theta,x)=\frac{e^{\kappa\cos(\theta-\mu)}}{2\pi I_{0}(\kappa)}\frac{1}{\sqrt{2\pi}\sigma_{c}}e^{-\frac{(x-\mu_{c})^{2}}{2\sigma_{c}^{2}}},
\]
clipping values of the $X$-coordinate between the bounded interval
$[0,2\pi]$. This distribution has the parameters $\mu\in[-\pi,\pi],\mu_{0}\in\mathbb{R},\kappa\geq0,\rho_{1}\in[0,1),\rho_{2}\in[0,1),\sigma>0$,
where $\mu,\kappa$ denote parameters for the (circular) marginal
along the $\Theta$-coordinate. and given $\Theta=\theta$, $X$ is
sampled from $N(\mu_{c},\sigma_{c}^{2})$, with

\begin{eqnarray*}
\mu_{c} & = & \mu+\sqrt{\kappa}\sigma\left\lbrace \rho_{1}(\cos\theta-\cos\mu)+\rho_{2}(\sin\theta-\sin\mu)\right\rbrace ,\\
\sigma_{c} & = & \sigma^{2}(1-\rho^{2}),\rho=(\rho_{1}^{2}+\rho_{2}^{2})^{1/2}.
\end{eqnarray*}
In our experiments, we fix $\rho_{1}=\rho_{2}=0.5,\sigma=1,\kappa=2$
across both populations. As random samples of distributions, we draw
$\mu,\mu_{0}\sim\text{{Unif}}(0,1)$ and $\mu,\mu_{0}\sim\text{{Unif}}(\delta,\delta+1)$
for samples of group 1 and 2 respectively, with $n_{1}=60,n_{2}=40$.
We repeat the above for $\delta\in[0,30]$ degrees converted to radians,
and obtain bivariate histograms corresponding to each sample distribution
from 500 random draws from that distribution. To evaluate the effects
of choosing $L,D'$ we calculate our embeddings for $L\in\lbrace2,3,4,5\rbrace,D'\in\lbrace6,8,12,24,48\rbrace$.
The choice of $L$ has small effect on performance, so we report results
for $L=4$ in Figure \ref{fig:power} (d). Higher values of $D'$
result in some increase in power.

\paragraph*{Discussion}
Our $IS\W$-based method is able to exploit the non-euclidean nature
of the problems and and their generality beyond mean comparison more
effectively than competing methods, which are based on mean comparison
on functional data/densities (frechet ANOVA, all functional ANOVA
methods), and/or L2 distance-based comparisons (all functional ANOVA
methods, DISCO). Regarding the optimal choice of embedding dimensions,
while proving theorem \ref{Thm:asy-tilde} we show that (Proposition
10 therein) choosing both $L$ and $D$ above certain thresholds ensures close approximation to the population
test statistic. For the combination test, adding more dimensions to
the embedding can have a two-fold effect: a) probing more dimensions
can help with finding differences, but b) every dimension adds another
test and so potentially leads to loss of power. Thus, for the combination
test, there must be an optimal data dependent choice of the embedding
dimension, which can potentially be found via split testing procedures.
We leave this to future work.

\paragraph*{Computational Complexity}
Assume an underlying graph $G=(V,E)$, and we use $L$ slices, $D'$ quantiles to calculate $IS\W$. Each distribution consists of $V$ atoms (distribution lives on graph vertices). Then, computation of our embeddings includes the following steps:

\begin{enumerate}
\item The computation of eigenvectors/values of symmetric sparse matrices is a well-studied problem with stable and efficient algorithms available, e.g. ARPACK providing an implementation of the Arnoldi method in MATLAB, Python, or R. These methods incur linear time complexity in the size of graph: $O(L(|V|+|E|)$. This computation is done only once per domain, and the overhead is negligible ($<<1$ second) for our graph experiments.

\item Hilbert embedding computations require computing $L$ pushforwards and $D'$ quantiles, which need sorting. This complexity is $O(L(|V|\log |V|+|D'|)$.

\item Testing using p-value combination requires computing t-test p-values for $LD'$ dimensions, with overall complexity of $O(LD'N)$ where $N$ is total sample size.
\end{enumerate}

In contrast, Sobolev Transport (ST) requires computing the Sliced ST Distance, based on computing shortest paths from randomly selected nodes in a graph. This has complexity $O(k(|V|\log |V|+|E|)$ \citep[pg. 4]{sobolev} Implementing permutation tests means repeating the above computation $P$ times, with 
$P=1000$ being a typical choice.

\subsection{NHANES data on physical activity monitoring}
\label{subsec:app_nhanes}
As our first real data application, we analyze the Physical Activity
Monitor (PAM) data from the 2005-2006 National Health and Nutrition
Examination Survey (NHANES)\footnote{\url{https://wwwn.cdc.gov/Nchs/Nhanes/2005-2006/PAXRAW_D.htm}}.
This contains physical activity pattern readings for a large number
of people collected over 1 week period on a per-minute granularity.
After basic pre-processing steps to ensure no missing entries, as
well as data reliability and well-calibrated activity monitors, we
use data from 6839 individuals. The data for each individual corresponds
to device intensity value from the PAM for $24\times60=1440$ minutes
throughout the day, for 7 days.

\setlength{\tabcolsep}{1pt}

\begin{figure}[t]
\begin{centering}
\begin{tabular}{cccc}
\includegraphics[width=0.3\textwidth]{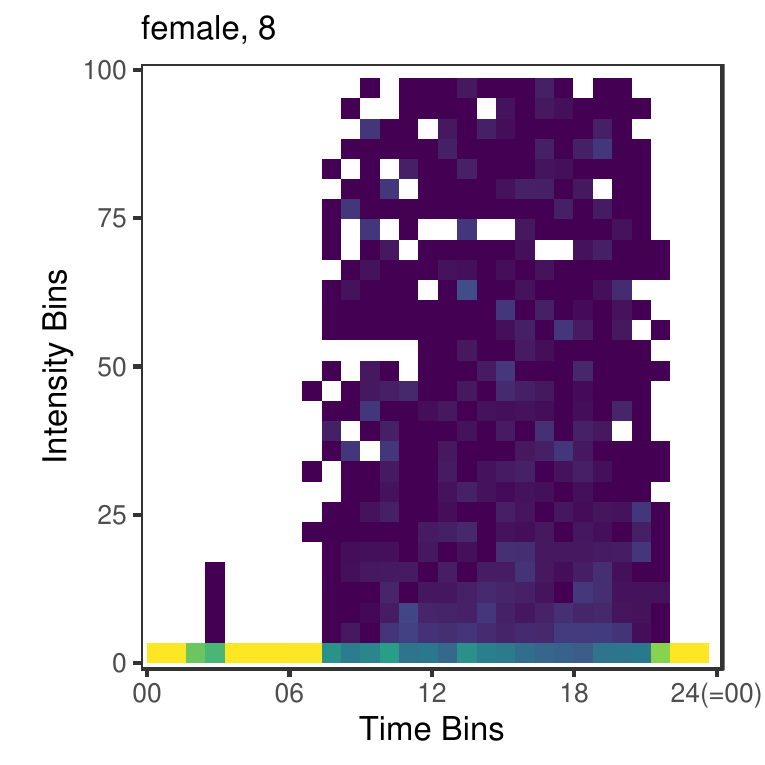} & \includegraphics[width=0.3\textwidth]{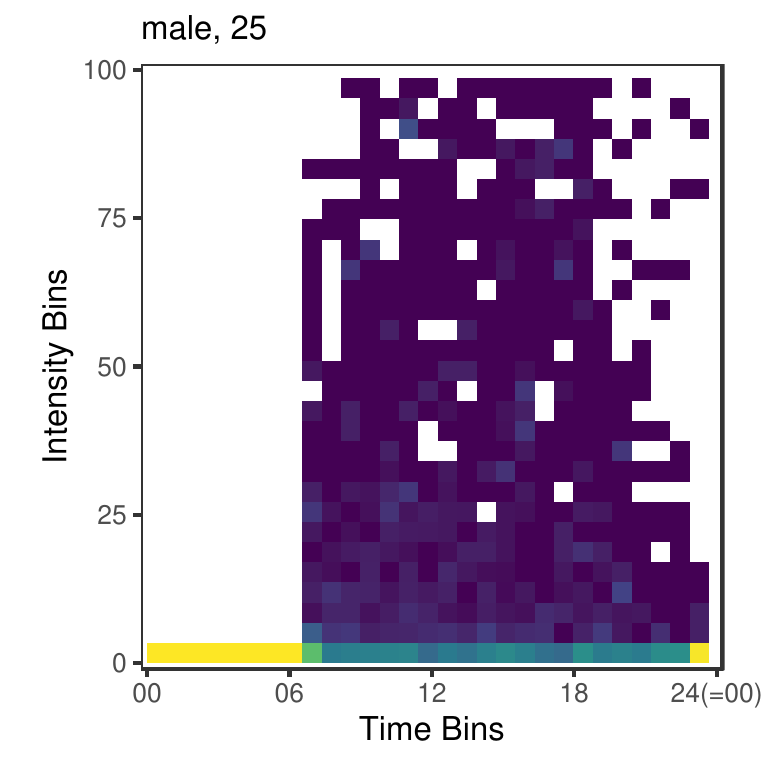} & \includegraphics[width=0.3\textwidth]{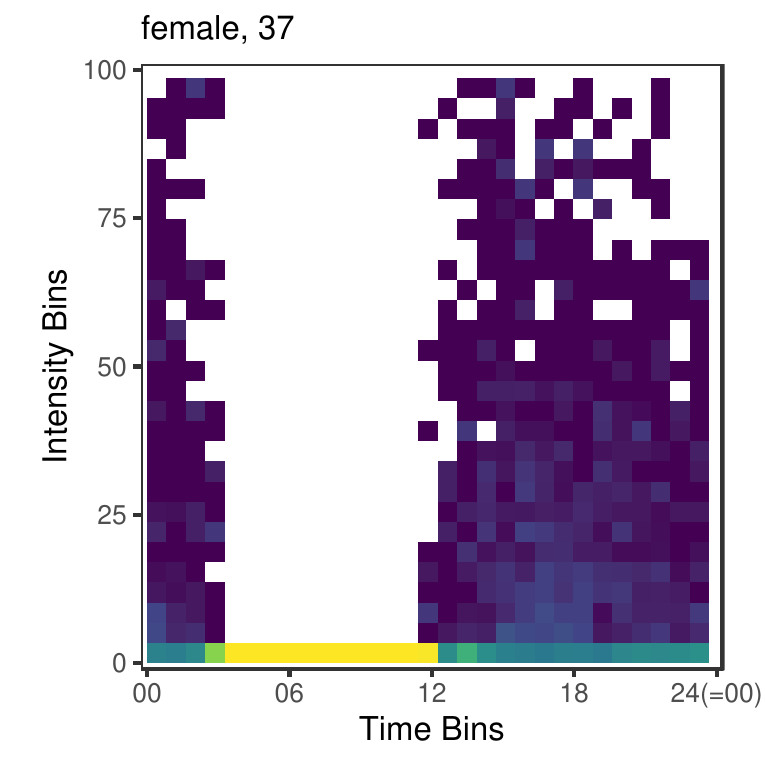} & \includegraphics[width=0.06\textwidth]{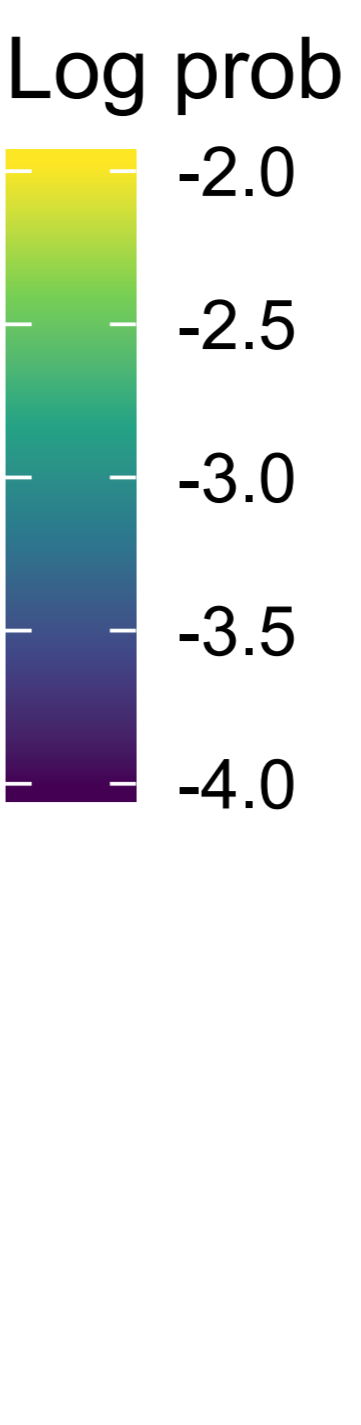}\tabularnewline
\end{tabular}
\par\end{centering}
\caption{\label{fig:nhanes_examples}Activity histograms for three individuals
from NHANES dataset. There are 100 bins in the intensity and 96 in
the time dimension; we show hour of day on the time axis. The time
dimension is periodic where 00:00 is identified with 24:00, giving
rise to a cylindrical histogram domain.}
\end{figure}

\begin{figure}[t]
\begin{centering}
\begin{tabular}{cccc}
\includegraphics[width=0.3\textwidth]{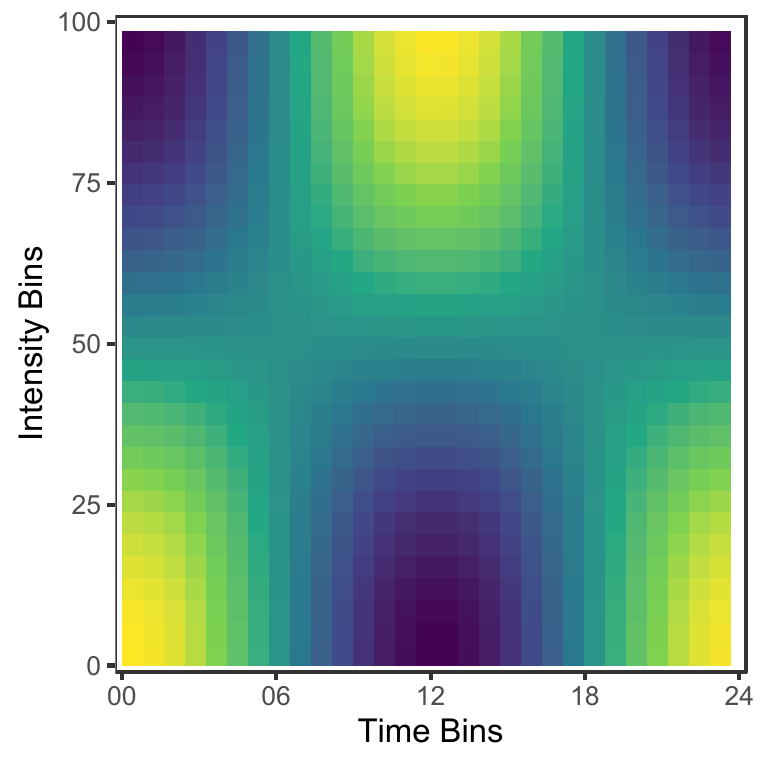} & \includegraphics[width=0.3\textwidth]{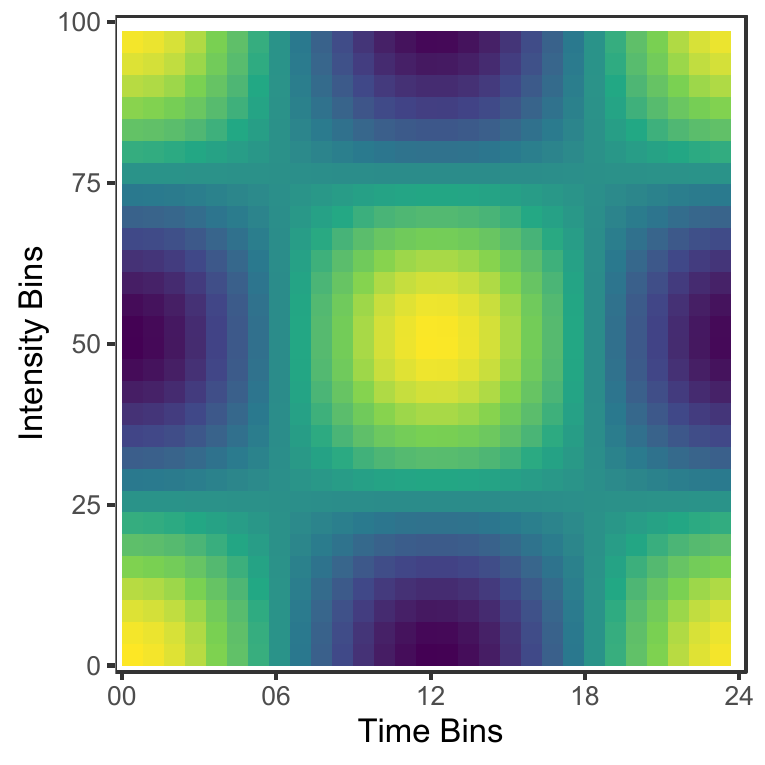} & \includegraphics[width=0.3\textwidth]{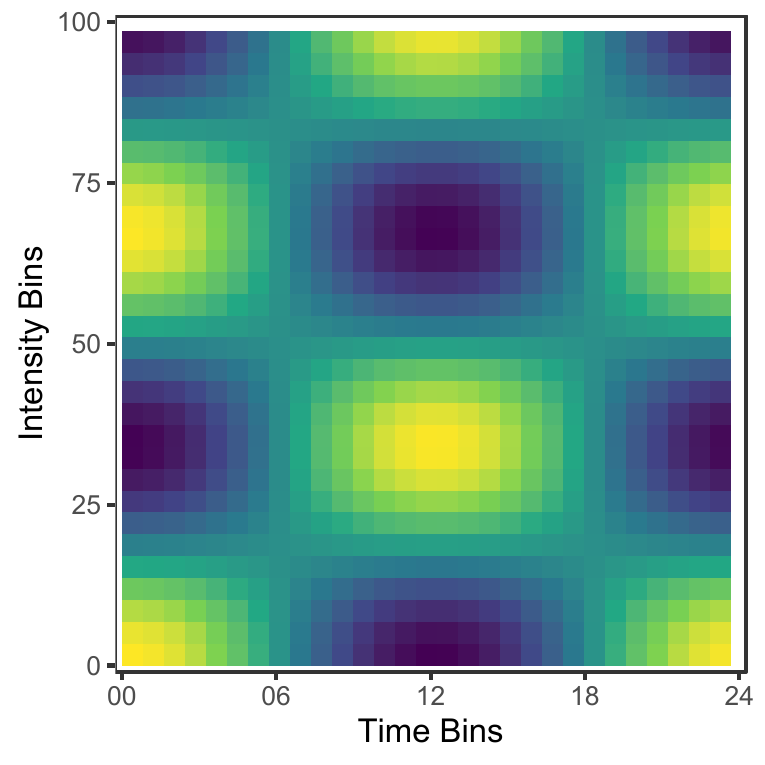} & \includegraphics[width=0.06\textwidth]{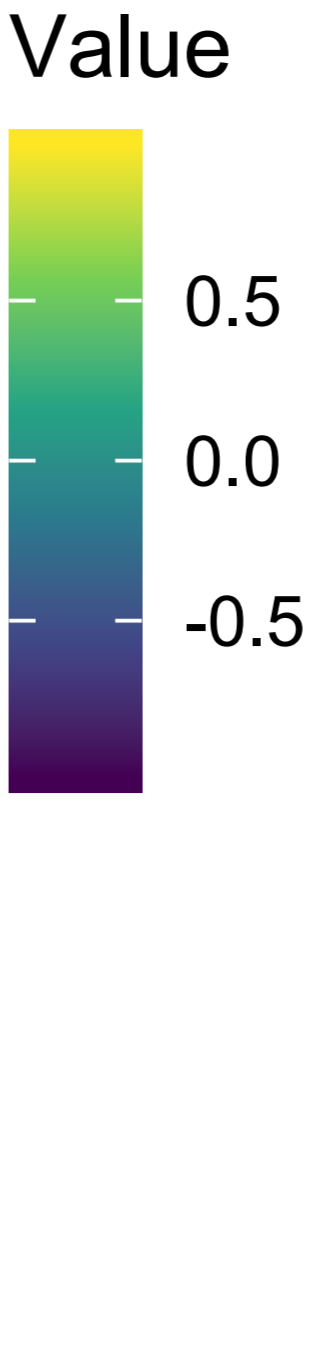}\tabularnewline
\end{tabular}
\par\end{centering}
\caption{\label{fig:nhanes_eigenfunctions}Three eigenfunctions for the NHANES
histogram domain normalized by the maximum absolute value. Note that
the eigenfunctions are periodic in the time direction (i.e. match
when glued over the side cut) but not in the intensity direction,
reflecting the cylindrical geometry of the underlying domain.}
\end{figure}

For each individual we can capture their activity patterns into a
cylindrical histogram with time and intensity dimensions. For each
observation, its time during the day is discretized into 15-minute
intervals giving 96 bins for the time dimension; its intensity value
(capped at 1,000) is discretized into a 100 equidistant bins. Since
the time dimension is periodic, we obtain a histogram over the cylinder
$S^{1}(T_{1})\times[0,T_{2})$, with $T_{1}=96,T_{2}=100$. Normalized
counts can thus be considered as person-specific probability distributions;
several examples are shown in Figure \ref{fig:nhanes_examples}. Note
that flattening the domain by cutting the cylinder will arbitrarily
split activity patterns (see especially Figure \ref{fig:nhanes_examples},
Female 37) and will lead to inefficiencies due to horizontal variability.

We apply the proposed methodology to check if the activity patterns
vary across different groups of individuals obtained as follows. We
first split the overall dataset based on the individual's age using
the following inclusive ranges: 6--15, 16--25, ...,76--85; this
covers all the ages in the dataset. From each split we sample 100
males and 100 females to avoid gender imbalance driving the results.
Thus, we end up with 8 age groups with 200 individuals per group.
Our goal is to compare these 8 groups' activity patterns by conducting
pair-wise tests. 

To perform our analysis we compute the eigenvalues and eigenfunctions
as per the 4th row of Table \ref{tab:Eigenvalues-and-eigenfunctions}
using $\ell_{1}=1,2,3$ and $\ell_{2}=1,2,3$, giving a total of $L=2\times3\times3=18$
eigenfunctions; three of the resulting eigenfunctions are shown in
Figure \ref{fig:nhanes_eigenfunctions}. We consider a $D'=5$ dimensional
embedding for the inverse CDF transformation, hence the final embedding
dimension after the slicing construction is $D=LD'=18\times5=90$. 

We summarize the results in Table \ref{tab:nhanes-age-groups}, \emph{below
the diagonal}. The $p$-values are obtained via the harmonic mean
combination approach. We run the Benjamini-Hochberg \cite{FDR} procedure
on the resulting $p$-values at the false discovery rate of $0.1$,
and the rejected hypotheses are indicated by the $p$-values in bold.
Our method detects statistically significant differences between all
pairs of groups, except 46--55 versus 36--45 and 56-65 groups. As
a control experiment, we provide our method with null cases and display
the $p$-values in Table \ref{tab:nhanes-age-groups}, \emph{above
the diagonal}. The null cases are obtained by combining the individuals
from the two comparison groups and splitting it arbitrarily (i.e.
mixing the two age groups). As expected, the $p$-values of the control
comparisons do not concentrate near zero.

\setlength{\tabcolsep}{1pt}

\begin{table}[t]
\begin{centering}
\begin{tabular}{ccccccccc}
\toprule 
Age Groups & 6--15 & 16--25 & 26--35 & 36--45 & 46--55 & 56--65 & 66--75 & 76--85\tabularnewline
\midrule 
6--15 &  & 0.979  & 0.31  & 0.383  & 0.297  & 0.905  & 0.921  & 0.326 \tabularnewline
16--25 & \textbf{3.7e-11 } &  & 0.998  & 0.963  & 0.443  & 0.872  & 0.442  & 0.529 \tabularnewline
26--35 & \textbf{4.6e-20 } & \textbf{1.0e-05 } &  & 0.987  & 0.818  & 0.93  & 0.731  & 0.992 \tabularnewline
36--45 & \textbf{3.2e-26 } & \textbf{3.5e-11 } & \textbf{0.01 } &  & 0.945  & 0.984  & 0.974  & 0.327 \tabularnewline
46--55 & \textbf{6.6e-27 } & \textbf{8.4e-16 } & \textbf{0.002 } & 0.377  &  & 0.832  & 0.618  & 0.844 \tabularnewline
56--65 & \textbf{2.4e-32 } & \textbf{7.5e-20 } & \textbf{3.1e-04 } & \textbf{0.042}  & 0.977  &  & 0.509  & 0.98 \tabularnewline
66--75 & \textbf{5.4e-45 } & \textbf{1.6e-16 } & \textbf{7.7e-06 } & \textbf{1.6e-04 } & \textbf{0.001 } & \textbf{0.011 } &  & 0.557 \tabularnewline
76--85 & \textbf{3.4e-52 } & \textbf{1.4e-23 } & \textbf{1.4e-15 } & \textbf{2.7e-12 } & \textbf{9.7e-16 } & \textbf{1.4e-09 } & \textbf{2.1e-06 } & \tabularnewline
\bottomrule
\end{tabular}
\par\end{centering}
\centering{}\caption{\label{tab:nhanes-age-groups} Comparing the activity intensity of
different age groups based on the NHANES dataset. Below diagonal:
$p$-values corresponding to the actual data comparisons. Above diagonal:
null $p$-values obtained by combining and randomly splitting the
two involved groups. The entries in boldface correspond to the rejected
hypotheses with the BH procedure at the FDR level of 0.1.}
\end{table}

Curiously, our method can be used ``off-label'' to conduct \emph{functional
data analyses} over different dimensions of the NHANES dataset. For
example, one can concentrate on a single day of activity intensity
data which gives a curve over the 24-hour circle. Since activity intensity
is a non-negative number, these curves can be normalized so as to
obtain probability distributions. Now we can use our methodology to
detect pair-wise differences across groups. While this has the benefit
of accounting for underlying geometry of data, it loses the absolute
magnitude information due to the normalization. Clearly the appropriateness
of such an analysis would depend on the goal of the exercise and the
particular research question attached to that goal; our proposal provides
a framework that is flexible enough to handle data of different modalities.

\subsection{Chicago Crime}
\label{subsec:app_chic}

\begin{figure}
\begin{centering}
\begin{tabular}{ccc}
\includegraphics[width=0.3\textwidth]{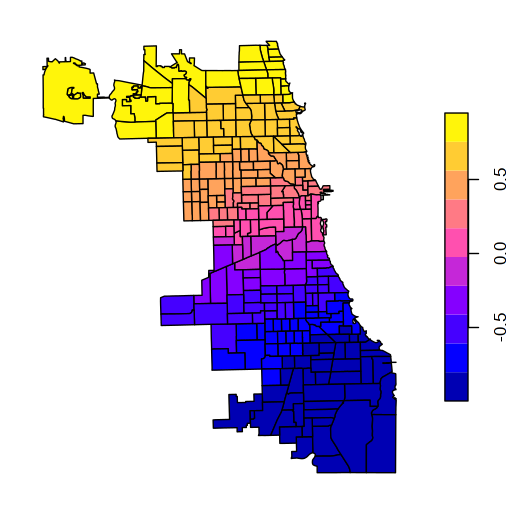} & \includegraphics[width=0.3\textwidth]{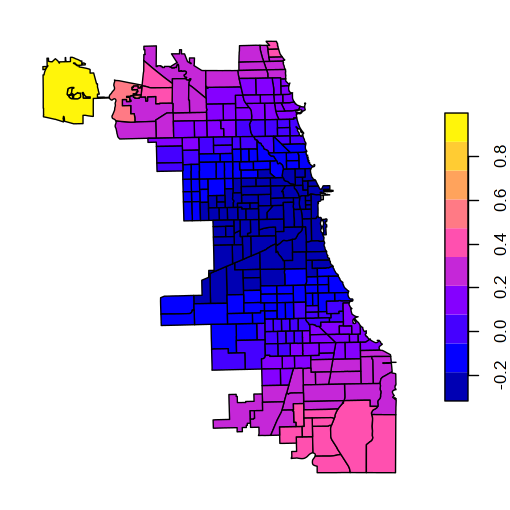} & \includegraphics[width=0.3\textwidth]{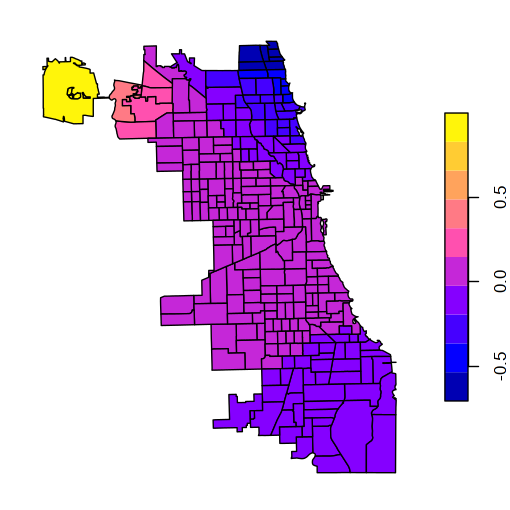}\tabularnewline
(a) $\phi_{1}(\cdot)$ & (b) $\phi_{2}(\cdot)$ & (c) $\phi_{3}(\cdot)$\tabularnewline
\end{tabular}
\par\end{centering}
\caption{\label{fig:beat-map}First three eigenvectors of the Laplacian are
shown for the beat adjacency graph, mapped back to the geographic
locations. All of the eigenvectors are normalized by the maximum absolute
value. The spatial smoothness of the eigenvectors---somewhat masked
here due to the discrete colormap---is crucial to efficiently capturing
horizontal variability of the data (i.e. distribution shifts over
the graph). The boundaries of beats are shown based on the shape file
from Chicago Data portal.}
\end{figure}

We demonstrate the use of our methodology on histograms over graphs.
In this experiment, we use the Chicago Crimes 2018 dataset\footnote{\url{data.cityofchicago.org}}
which captures incidents of crime in the City of Chicago. We base
our analysis on the type of crime, the beat (geographic area subdivision
used by police, see Figure \ref{fig:beat-map}) where the incident
took place, and the date of the incident. To capture the spatial aspect
of the data we build a graph with one vertex per beat; two vertices
are connected by an edge if the corresponding beats share a geographic
boundary. For each crime type and day, we capture the total count
of that crime type for each beat; after normalizing this gives a daily
probability distribution over the graph. Our goal is to compare the
collection of distributions of, say, theft occurring on Tuesday to
those of Thursday and Saturday. The Tuesday versus Thursday comparison
is intended as a null case, as we do not expect to see any differences
between them \cite{AKME}. 

We build the un-normalized Laplacian of the beat adjacency graph,
and compute its lowest frequency $L=20$ eigenvalues and eigenvectors.
The first three eigenvectors are plotted in Figure \ref{fig:beat-map}.
The number of inverse CDF values used in the embedding is $D'=5$,
which gives rise to $D=100$ dimensional embedding. The results of
comparisons are shown in the last two columns of Table \ref{tab:chicago-crime};
the $p$-values are obtained via the harmonic mean combination approach.
We run the Benjamini-Hochberg \cite{FDR} procedure on the 20 resulting
$p$-values at the false discovery rate of $0.1$, and the rejected
hypotheses are indicated by the $p$-values in bold. As expected,
no differences were detected between Tuesday and Thursday patterns.
On the other hand, we see that there are statistically significant
differences between Tuesday and Saturday patterns in the following
categories of crime: theft, deceptive practice, battery, robbery,
narcotics, and criminal damage. 

\setlength{\tabcolsep}{1pt}

\begin{table}[t]
\begin{centering}
\begin{tabular}{lcccccccc}
\toprule 
\multirow{2}{*}{Crime Type} & \multicolumn{2}{c}{Tuesday} & \multicolumn{2}{c}{Thursday} & \multicolumn{2}{c}{Saturday} & Tue vs Thu & Tue vs Sat\tabularnewline
\cmidrule{2-9} \cmidrule{3-9} \cmidrule{4-9} \cmidrule{5-9} \cmidrule{6-9} \cmidrule{7-9} \cmidrule{8-9} \cmidrule{9-9} 
 & $N$ & $\mathrm{\overline{count}}$ & $N$ & $\mathrm{\overline{count}}$ & $N$ & $\mathrm{\overline{count}}$ & $p$-value & $p$-value\tabularnewline
\midrule 
Theft  & 52  & 178.7  & 52  & 182.9  & 52  & 180.2  & 0.452  & \textbf{4.7e-06 }\tabularnewline
Deceptive Practice  & 51  & 55.8  & 52  & 54.9  & 52  & 44.4  & 0.255  & \textbf{4.2e-04 }\tabularnewline
Battery  & 52  & 125.8  & 52  & 123.0  & 52  & 154.9  & 0.374  & \textbf{0.001 }\tabularnewline
Robbery  & 50  & 25.2  & 50  & 25.1  & 52  & 28.1  & 0.130  & \textbf{0.002 }\tabularnewline
Narcotics  & 51  & 36.0  & 51  & 34.6  & 50  & 36.9  & 0.890  & \textbf{0.008 }\tabularnewline
Criminal Damage  & 52  & 70.0  & 52  & 73.7  & 52  & 83.0  & 0.901  & \textbf{0.03 }\tabularnewline
Other Offense  & 52  & 49.5  & 52  & 48.4  & 52  & 44.1  & 0.670  & 0.037 \tabularnewline
Burglary  & 52  & 34.0  & 52  & 33.1  & 52  & 29.1  & 0.157  & 0.183 \tabularnewline
Motor Vehicle Theft  & 52  & 27.9  & 52  & 26.2  & 51  & 28.1  & 0.923  & 0.365 \tabularnewline
Assault  & 52  & 57.2  & 52  & 59.3  & 52  & 52.4  & 0.996  & 0.617 \tabularnewline
\bottomrule
\end{tabular}
\par\end{centering}
\caption{\label{tab:chicago-crime}Results on Chicago Crime 2018 dataset. The
entries in bold correspond to the rejected hypotheses with the BH
procedure at the FDR level of 0.1. The $N$ column captures the number
of days passing the filtering criteria, and the $\mathrm{\overline{count}}$
column shows the average per-day crime count.}
\end{table}

\subsection{Brain Connectomics}
\label{subsec:app_brain}
In this example, we consider two publicly available brain connectomics
datasets \cite{func_conn_data1,func_conn_data2} distributed as a
part of the R package \texttt{graphclass}\footnote{\url{http://github.com/jesusdaniel/graphclass}}.
Both are based on resting state functional magnetic resonance imaging
(fMRI): COBRE has data on 54 schizophrenics and 70 controls, and UMich
with 39 schizophrenics and 40 controls. The datasets capture the pairwise
correlations between 264 regions of interest (ROI) of Power parcellation
\cite{Power2011} and can be considered as a 264 node graph (263 nodes
for COBRE as ROI 75 is missing) with positive and negative edge weights.

We define three probability measures supported on the nodes of the
graph. For each ROI we take the sum of absolute values of all its
correlations with the remaining ROIs. Now we have a positive number
assigned to each node capturing its overall connectivity to the rest
of the graph and we normalize to obtain a measure; this construction
will be referred to as ``all correlations''. Note that each scanned
subject gives rise to a separate ``all correlations'' probability
measure on the same underlying node set. The ``positive correlations''
and ``negative correlations'' constructions are based on keeping
respectively only positive or only negative correlations and aggregating
as above.

We also need a fixed base graph for the computation of the Laplacian
eigen-decomposition; this graph should capture the spatial connectivity
of the ROIs which is relevant due to the smooth nature of the blood
oxygenation level dependent (BOLD) signal that is used for computing
the correlations. To this end, we obtain the coordinates for the centers
of the 264 ROIs\footnote{\url{www.jonathanpower.net/2011-neuron-bigbrain.html}}
and build the base graph by connecting each ROI to its nearest 8 ROIs.We
compute the lowest frequency $L=20$ eigenvalues and eigenvectors
of the corresponding un-normalized Laplacian. The number of inverse
CDF values used in the embedding is $D'=5$, which gives rise to $D=100$
dimensional embedding.

\begin{table}
\begin{centering}
\begin{tabular}{cccc}
\toprule 
\multirow{1}{*}{Dataset} & \multicolumn{1}{c}{All correlations} & \multicolumn{1}{c}{Positive correlations} & \multicolumn{1}{c}{Negative correlations}\tabularnewline
\midrule 
COBRE & 0.0084 & 0.00019 & 0.0019\tabularnewline
UMich & 0.609 & 0.116 & 0.022\tabularnewline
\bottomrule
\end{tabular}
\par\end{centering}
\caption{\label{tab:Results-on-brain}Comparison results between the schizophrenic
and control groups for brain connectomics datasets.}
\end{table}

Table \ref{tab:Results-on-brain} shows the result of comparing the
schizophrenic group to the control group for both of the datasets;
the $p$-values are obtained via the harmonic mean combination approach.
We can see that our approach detects statistically significant differences
between the two groups in COBRE dataset in all of the three types
of measures on graphs. In contrast, for UMich dataset, the difference
is detected only in the negative correlations and loses significance
when corrected for multiple testing. This is potentially caused by
the higher inhomogeneity of the UMich dataset that was pooled across
five different experiments spanning seven years \cite{func_conn_data2}.
An interesting aspect of our analysis is that due to normalization
(to obtain probability measures) the total sum of connectivity is
factored out by the proposed method. As a result, the detected differences
are not related to the well-known change in the overall connectivities
between the two groups, but rather to distributional changes in marginal
connectivity strengths.

% \bibliographystyle{plainnat}
% \bibliography{biblio}

% \section{Frequently Asked Questions}
% \label{sec:appC}

% %%%%%%%%%%%%%%%%%%%%%%%%%%%%%%%%%%%%%%%%%%%%%%%%%%%%%%%%%%%%